\pdfoutput=1
\documentclass[%
 reprint,
 amsmath,amssymb,
 aps,
 nofootinbib,
 superscriptaddress
]{revtex4-2}

\usepackage{amsthm}
\usepackage{graphicx}
\usepackage{xcolor}
\usepackage{braket}
\usepackage{amsmath}
\usepackage{amssymb}
\usepackage{enumerate}
\usepackage{booktabs}
\usepackage{caption}
\usepackage{mathtools}
\usepackage{etoolbox}
\usepackage{algorithmic}
\usepackage{algorithm}
\usepackage{bbm}
\usepackage{here}

\usepackage[
colorlinks=true,
urlcolor=blue,
citecolor=blue,
linkcolor=blue,
hyperfootnotes=false,
breaklinks=true
]{hyperref}

\newtheorem{theorem}{Theorem}
\newtheorem{lemma}{Lemma}

\newtheorem{definition}{Definition}
\newtheorem{assumption}{Assumption}

\allowdisplaybreaks

\begin{document}

\title{Quantum Metropolis-Hastings algorithm with the target distribution calculated by quantum Monte Carlo integration}
\author{Koichi Miyamoto}
\email{miyamoto.kouichi.qiqb@osaka-u.ac.jp}
\affiliation{Center for Quantum Information and Quantum Biology, Osaka University, Toyonaka, Osaka 560-0043, Japan}

\date{\today}

\begin{abstract}

The Markov chain Monte Carlo method (MCMC), especially the Metropolis-Hastings (MH) algorithm, is a widely used technique for sampling from a target probability distribution $P$ on a state space $\Omega$ and applied to various problems such as estimation of parameters in statistical models in the Bayesian approach.
Quantum algorithms for MCMC have been proposed, yielding the quadratic speedup with respect to the spectral gap $\Delta$ compered to classical counterparts.
In this paper, we consider the quantum version of the MH algorithm in the case that calculating $P$ is costly because the log-likelihood $L$ for a state $x\in\Omega$ is obtained via computing the sum of many terms $\frac{1}{M}\sum_{i=0}^{M-1} \ell(i,x)$.
We propose calculating $L$ by quantum Monte Carlo integration and combine it with the existing method called quantum simulated annealing (QSA) to generate the quantum state that encodes $P$ in amplitudes.
We consider not only state generation but also finding a credible interval for a parameter, a common task in Bayesian inference.
In the proposed method for credible interval calculation, the number of queries to the quantum circuit to compute $\ell$ scales on $\Delta$, the required accuracy $\epsilon$ and the standard deviation $\sigma$ of $\ell$ as $\tilde{O}(\sigma/\epsilon^2\Delta^{3/2})$, in contrast to $\tilde{O}(M/\epsilon\Delta^{1/2})$ for QSA with $L$ calculated exactly.
Therefore, the proposed method is advantageous if $\sigma$ scales on $M$ sublinearly.
As one such example, we consider parameter estimation in a gravitational wave experiment, where $\sigma=O(M^{1/2})$.
\end{abstract}

\maketitle

\section{Introduction}

\begin{table*}[tp]
\caption{The complexities (number of queries to the oracle to compute $\ell$) in various tasks in various methods for sufficiently small $\epsilon$ (precisely speaking, $\epsilon$ satisfying Eq. (\ref{eq:smallEpsCond})).
Here, $\epsilon$ represents the total variation distance between the generated state and $\ket{P}$ and the error in the cumulative distribution function for state generation and credible interval calculation, respectively (see Secs. \ref{sec:appQSA} and \ref{sec:CI} for the detail).
$\Delta_{\rm min}$ is a lower bound of spectral gaps of some Markov chains (see Sec. \ref{sec:appQSA} for the detail).
$\rho$ is the SNR in GW matched filtering (see Sec. \ref{sec:GW} for the detail).}
 \label{tbl:CompSum}
 \centering
  \begin{tabular}{c|ccc}
   \hline
   Task & Proposed method & Exact QSA & Classical MH \\
   \hline
   Generate $\ket{P}$ & $\tilde{O}\left(\frac{\sigma \bar{L}^{1/2}}{ \Delta^{3/2}_{\rm min}\epsilon}\right)$ & $\tilde{O}\left(\frac{M\bar{L}^{1/2}}{\Delta_{\rm min}^{1/2}}\right)$ & Not applicable \\
   Credible interval (general) & $\tilde{O}\left(\frac{\sigma \bar{L}^{1/2}}{ \Delta_{\rm min}^{3/2}\epsilon^2}\right)$ & $\tilde{O}\left(\frac{M\bar{L}^{1/2}}{\Delta_{\rm min}^{1/2}\epsilon}\right)$ & $\tilde{O}\left(\frac{M}{\Delta\epsilon^2}\right)$ \\
   Credible interval (GW) & $\tilde{O}\left(\frac{\rho M^{1/2}  \bar{L}^{1/2}}{ \Delta_{\rm min}^{3/2}\epsilon^2}\right)$ & $\tilde{O}\left(\frac{M\bar{L}^{1/2}}{\Delta_{\rm min}^{1/2}\epsilon}\right)$ & $\tilde{O}\left(\frac{M}{\Delta\epsilon^2}\right)$ \\
   \hline
  \end{tabular}
\end{table*}

Following the recent rapid development of quantum computing, various quantum algorithms are studied extensively, along with their industrial and scientific applications.
Among them, quantum algorithms for the Markov Chain Monte Carlo method (MCMC) are one of prominent ones \cite{Somma0712.1008,Somma2008,Wocjan2008,Yung2012,Harrow2020,Lemieux2020efficientquantum,campos2022quantum}.
MCMC is a methodology for sampling from a probability distribution $P$ on a sample space (state space) $\Omega$ by generating a Markov chain whose stationary distribution is $P$ (see \cite{levin2017markov} as a textbook).
It is widely used in various situations, for example, estimation of parameters in statistical models in the Bayesian approach \cite{gelman1995bayesian}. 

In spite of its usefulness, MCMC often has an issue of computational time, since it can be needed to make many iterations of state generation for the chain to sufficiently converge to the target distribution $P$.
Classically, the iteration number for sufficient convergence scales as $\tilde{O}(\Delta^{-1})$, where $\Delta$ is the spectral gap of the chain (see the definition in Sec. \ref{sec:MH}).

Quantum MCMC algorithms can be remedies for this: using the quantum walk \cite{Szegedy2004} as a building block, they generates (an approximation of) a quantum state $\ket{P}$ that encodes $P$ in amplitudes with $\tilde{O}(\Delta^{-1/2})$ queries to the walk operator, which indicates the quadratic quantum speedup.
More concretely, the approach in \cite{Somma0712.1008,Somma2008,Wocjan2008,Yung2012,Harrow2020} called quantum simulated annealing (QSA) takes the following strategy.
We consider an initial distribution $P_0$ and a series of Markov chains with stationary distributions $P_1,...,P_l$ such that these distributions sufficiently {\it overlap}, that is, the quantum states that encode them satisfy $|\braket{P_i | P_{i+1}}|^2\ge {\rm const}$.
Then, starting from the state $\ket{P_0}$ that encodes $P_0$, we sequentially generate $\ket{P_1},...,\ket{P_{l-1}}$, and finally $\ket{P_{l}}$, which is close to $\ket{P}$.
In particular, \cite{Harrow2020} considered applying this to Bayesian inference, where $P$ is written as $P(x)\propto P_0(x) e^{-L(x)}$ with the negative log-likelihood $L$, and presented the procedure to generate $\ket{P}$ making $\tilde{O}(\sqrt{\bar{L}/\Delta})$ calls to the quantum walk operator, where $\bar{L}$ is the expectation of $L$ under the distribution $P_0$.

Among various types of MCMC, the Metropolis-Hastings (MH) algorithm \cite{metropolis1953equation,Hastings1970} is particularly prominent.
In this algorithm, it is supposed that the target distribution $P$ is efficiently computable expect for the normalization constant and we are given some proposal distribution $T$ for transition among possible states, which is also efficiently computed.
Then, accepting the proposed transition with some probability determined by $P$ and $T$, we generate a chain, which is guaranteed to converge to $P$.
Because of its simplicity, the MH algorithm is widely used.
Some of the previous quantum algorithms for MCMC are in fact based on the MH algorithm \cite{Yung2012,Lemieux2020efficientquantum,campos2022quantum}.

In this paper, we consider a quantum Metropolis-Hastings algorithm in a specific but ubiquitous and important situation.
That is, we focus on the case that the target distribution $P$ is computed via summation of many terms.
More specifically, we suppose that $P(x)\propto P_0(x) e^{-L(x)}$ and $L$ is written as $L(x)=\frac{1}{M}\sum_{i=0}^{M-1} \ell(i,x)$ with a large integer $M$ and a function $\ell$, except for efficiently computable terms (see Sec. \ref{sec:walkOpQMCI} for the exact problem setting).
In the context of Bayesian inference, this corresponds to the situation that the negative log-likelihood is a sum of many terms.
In this case, calculation of $P$ can be time-consuming, even if $\ell$ is efficiently computed.
Naively, we need to iterate calculations of $\ell(i,x)$ and additions $M$ times.

As an example of such a case, we can take parameter estimation in a gravitational wave (GW) detection experiment (see a review \cite{thrane_talbot_2019} and references therein).
In a GW laser interferometer such as LIGO and Virgo \cite{GW150914,GWTC-1,GWTC-2,GWTC-3}, a GW signal is explored in noisy detector output data by matched filtering \cite{Balasubramanian1996,Owen1996,Owen1999,Allen2012}, and, if detected, estimation of parameters in the waveform of GW is performed in the Bayesian approach.
The log-likelihood in this parameter estimation is given as a sum of contributions from various Fourier modes of the data and thus has the aforementioned form.
More generally, many statistical inference problems with a large number of independent samples fall into the considered case.

As far as the author knows, previous studies on quantum MCMC algorithms have not focused on the difficulty to compute a target distribution of the above type.
\cite{Somma0712.1008,Somma2008,Wocjan2008,Yung2012,Harrow2020} assumed the availability of the quantum circuit to generate the quantum state that encodes the transition matrix of the Markov chain.
Although \cite{Lemieux2020efficientquantum,campos2022quantum} broke down the operators needed in the quantum MH algorithm in more details, they assumed the availability of the quantum circuit to determine acceptance or rejection, and did not consider the detail of target distribution calculation.
When it comes to quantum algorithms for data analysis in GW experiments, although there are studies on GW detection \cite{Gao2022,Miyamoto2022,veske2022quantum,postema2022hybrid} and the quantum MH algorithm for GW parameter estimation \cite{Escrig_2023}, the issue of log-likelihood computation has not been focused on.

Then, in this paper, we consider how to speedup calculation of $P$ using another quantum algorithm as a subroutine of the quantum MH algorithm.
Concretely, we use quantum Monte Carlo integration (QMCI) \cite{Montanaro2015}.
Based on quantum amplitude estimation (QAE) \cite{brassard2002quantum}, QMCI estimates $E[F(X)]$, the expectation of a function $F$ of a random variable $X$, providing quadratic speedup compared to classical Monte Carlo integration.
For example, if we have a bound $\sigma^2$ on the variance of $F(X)$, QMCI yields an estimate with accuracy $\epsilon$, calling the quantum circuits to generate a quantum state encoding $X$'s distribution and compute $F$ $\tilde{O}(\sigma/\epsilon)$ times.
As a special case, we can use QMCI to estimate the sum of many terms.
In fact, QMCI is utilized for calculating the signal-to-noise ratio (SNR) in the quantum algorithm for GW matched filtering proposed in \cite{Miyamoto2022}, and using QMCI in the quantum MH algorithm is a similar idea.

We note that QMCI gives an erroneous estimate and thus the Markov chain based on it has a stationary distribution different from the original one $P$.
Fortunately, MCMC with such a perturbation has been studied \cite{alquier2016noisy,medina2016stability,Rudolf2018}, and, according to their results, we can set the accuracy in QMCI to obtain the distribution close to $P$.

We make a further consideration from a practical perspective.
Some of previous studies considered only preparing the quantum state $\ket{P}$, but what we want in real life is not the quantum state but the results of some statistical analysis on $P$ as classical data.
Then, this paper presents not only how to prepare $\ket{P}$ but also the procedure to obtain the {\it credible interval} of a parameter $\theta$ in a statistical model.
It is an interval where $\theta$ falls with a prefixed probability in the distribution $P$ and a quantity we often aim to find in Bayesian inference.
Given a quantum circuit $V_P$ to prepare $\ket{P}$, we can estimate the cumulative distribution function (CDF) of $\theta$ by QMCI using $V_P$ iteratively.
We then find the credible interval via binary search on the CDF.
We also consider applying this method to credible interval calculation for GW parameters.

Table \ref{tbl:CompSum} is a summary of the complexities in state generation and credible interval calculation, which mean the numbers of calls to the oracle to compute $\ell$, in various methods for sufficiently small error tolerance $\epsilon$.
Here, the exact QSA method is QSA with $L$ calculated exactly by $M$ iterative calculations of $\ell$.
We see that the complexity of the proposed method is equal to that of the exact QSA method with the factor $M$ replaced with $\sigma/\Delta_{\rm min}\epsilon$, where $\sigma^2$ is the variance of $\ell$, $\epsilon$ is the required accuracy, and $\Delta_{\rm min}$ is the lower bound of spectral gaps of Markov chains related to the considered problem.
This is because, following the result of \cite{alquier2016noisy}, we estimate $L$ by QMCI with accuracy $\Delta_{\rm min}\epsilon$ in order to reach a distribution close to $P$ with accuracy $\epsilon$, and thus its complexity becomes $\tilde{O}(\sigma/\Delta_{\rm min}\epsilon)$.
As a result, with respect to $\epsilon$ and $\Delta_{\rm min}$, the complexity of the proposed method is worse than the exact QSA method and even the classical MH algorithm.
Nevertheless, it may be advantageous with respect to $M$, if $\sigma$ scales on $M$ sublinearly.
In fact, in the case of GW parameter estimation, $\sigma$ can be $O(\sqrt{M})$, which means that the proposed method is quadratically faster than the exact QSA method and the classical MH algorithm with respect to $M$.

The rest of this paper is organized as follows.
Sec. \ref{sec:prel} is preliminary one, where we briefly explain the MH algorithm, QSA and QMCI.
In Sec. \ref{sec:ourAlg}, we present our methods for generating the state $\ket{P}$ and finding the credible interval in details.
In Sec. \ref{sec:GW}, we consider the application to GW parameter estimation.
Sec. \ref{sec:summary} summarizes this paper.

\section{Preliminary \label{sec:prel}}

\subsection{Notation \label{sec:notation}}

Here, we summarize some notations used in this paper.

$\mathbb{R}_+$ denotes the set of all positive real numbers. 

For $n\in\mathbb{N}$, we define $[n]:=\{1,...,n\}$ and $[n]_0:=\{0,1,...,n-1\}$.

We hereafter consider systems consisting of quantum registers (or simply registers), sets of single or multiple qubits.
A ket $\ket{\psi}$ denotes a state vector of a quantum state on a register, and we sometime put a subscript to clarify the register on which the state is generated: $\ket{\psi}_R$ is a state on a register $R$.
Similarly, we sometime put a subscript to a symbol representing an operator to indicate the register on which the operator acts.
In particular, $I_R$ denotes the identity operator on a register $R$.

For $x\in\mathbb{R}$, $\ket{x}$ denotes the computational basis state on a register whose bit string corresponds to a finite-precision binary representation of $x$.
We assume that any number considered in this paper is represented with a sufficiently large number of qubits and thus neglect rounding errors.
For a real vector $x=(x_1,...,x_d)\in\mathbb{R}^d$, $\ket{x}$ denotes a computational basis state on $d$-register system $\ket{x}=\ket{x_1}\cdots\ket{x_d}$.

For a vector $x=(x_1,...,x_d)\in\mathbb{C}^d$, we define its $k$-norm as $\|x\|_k:=\left(\sum_{i=1}^d |x_i|^k\right)^{1/k}$ with $k\in\mathbb{N}$ and max norm as $\|x\|_\infty:=\max\{|x_1|,...,|x_d|\}$.
For a matrix $A\in\mathbb{C}^{m \times n}$, we define $\|A\|_k:=\sup_{\substack{x\in\mathbb{C}^n \\ \|x\|_k=1}} \|Ax\|_k$ with $k\in\mathbb{N}\cup\{\infty\}$, and denote by $\|A\|_F$ its Frobenius norm.
We simply write $\|\cdot\|_2$ as $\|\cdot\|$.
$\|\ket{\psi}\|$ is a 2-norm of the (unnormalized) state vector $\ket{\psi}$.

If $x,y\in\mathbb{C}^d$ satisfy $\|x-y\|\le\epsilon$ with some $\epsilon\in\mathbb{R}_+$, we say that $x$ is $\epsilon$-close to $y$ and $x$ is an $\epsilon$-approximation of $y$.
If quantum states $\ket{\psi}$ and $\ket{\phi}$ on a same register satisfy $\|\ket{\psi}-\ket{\phi}\|\le\epsilon$ with some $\epsilon\in\mathbb{R}_+$, we say that $\ket{\psi}$ is $\epsilon$-close to $\ket{\phi}$ and $\ket{\psi}$ is an $\epsilon$-approximation of $\ket{\phi}$.

For a nonsingular matrix $A$, we define its condition number as $\|A\|\cdot\|A^{-1}\|$.

Letting $(\Omega,2^\Omega,P)$ be a probability space with a finite sample space $\Omega$, we write $P(x)=P(\{x\})$ for $x\in\Omega$ and also call the measure $P$ the probability distribution or distribution.
We denote by $\mathbb{E}_P[\cdot]$ the expectation with respect to $P$. 

The indicator function $\mathbf{1}_C$ takes 1 if the condition $C$ is satisfied and 0 otherwise.

\subsection{Metropolis-Hastings algorithm \label{sec:MH}}

\subsubsection{Outline}

We briefly summarize the MH algorithm \cite{metropolis1953equation,Hastings1970}, whose aim is sampling a random variable $X$ that obeys some target probability distribution.
For more details, see \cite{levin2017markov}.

Every value $X$ can take is called a {\it state}, and the set of the states is called the {\it state space} and hereafter denoted by $\Omega$.
In this paper, we consider the situation that $\Omega$ is a finite subset in $\mathbb{R}^d$, where $d\in\mathbb{N}$.
This is because a quantum computer can only represent real numbers in finite precision using a finite number of qubits, which is the case also on a classical computer.
Of course, $X$ can take continuous values in many situations, but we assume that continuous $X$ is well approximated in a discrete manner and errors from this are negligible, as stated in Sec. \ref{sec:notation}.

For every $x\in\Omega$, we denote by $P(x)\in(0,1)$ the probability that $X$ takes $x$ in the target distribution.
We assume that $P(x)$ can be written as $P(x)=p(x)/Z$, where $p(x)$ is an easily computable function and $Z:=\sum_{x\in\Omega}p(x)$ is the normalization factor.
Although $Z$ is often hard to be computed, the MH algorithm works even if we do not know $Z$, as explained later.

In the MH algorithm, starting from some initial state $x_0$, we sequentially get states by making transitions over $\Omega$ as follows.
For every $x\in\Omega$, we set some easy-to-sample proposal distribution $T(x,\cdot):\Omega\rightarrow (0,1)$, for example, the normal distribution centered at $x$ (strictly, its discrete approximation).
Letting $x_i\in\Omega$ be the $i$th state, we randomly choose $\tilde{x}_{i+1} \in \Omega$ with probability $T(x_i,\tilde{x}_{i+1})$ as a candidate for the next state.
Then, calculating the acceptance ratio $A(x_i,\tilde{x}_{i+1})$, which is defined for $x,y\in\Omega$ as
\begin{equation}
    A(x,y):=
    \min\left\{1,\frac{P(y)T(y,x)}{P(x)T(x,y)}\right\}, \label{eq:accRatio}
\end{equation}
we set the next state $x_{i+1}$ to $\tilde{x}_{i+1}$ with probability $A(x_i,\tilde{x}_{i+1})$ or stay at $x_i$ otherwise.
Note that the target distribution $P$ appears in $A$ in the form of the ratio $P(y)/P(x)$, which means that we need to compute only $p(x)$, not $Z$.

As a consequence, the sequence generated by the MH algorithm is a Markov chain with a following transition matrix $W=(W_{x,y})$: it is indexed by $x,y\in\Omega$ and its $(x,y)$ entry, which corresponds to the probability that the transition to $y$ occurs provided that the current state is $x$, is
\begin{equation}
    W_{x,y}=
    \begin{cases}
    T(x,y)A(x,y) & {\rm if} \ x \ne y \\
    1-\sum_{z\in\Omega\setminus\{x\}} T(x,z)A(x,z) & {\rm if} \ x=y
    \end{cases}.
    \label{eq:tranMat}
\end{equation}

The convergence property of this Markov chain is affected by the {\it spectral gap} $\Delta$.
It is defined as
\begin{equation}
\Delta:=1-|\lambda_1|,
\end{equation}
where $\lambda_1$ the eigenvalue of $W$ with the second largest modulus.
It is known that the eigenvalue of $W$ with the largest modulus is 1 and it is non-degenerate \cite[LEMMAs 12.1 and 12.2]{levin2017markov}, which means $|\lambda_1|<1$ and $0<\Delta<1$.
To present the formal statement on the convergence rate, we introduce the total variance distance, a metric to measure the difference between two probability distributions.
\begin{definition}
For probability measures $P$ and $Q$ on a measurable space $(\Omega,\mathcal{F})$, the total variance distance is defined as
\begin{equation}
    \|P-Q\|_{\rm TV}:=\sup_{A\in \mathcal{F}}|P(A)-Q(A)|. \label{eq:TV}
\end{equation}
\end{definition}
\noindent Since we are now considering finite $\Omega$, Eq. (\ref{eq:TV}) becomes $\|P-Q\|_{\rm TV}=\max_{A\in 2^\Omega}|P(A)-Q(A)|$.

Then, we have the following theorem.

\begin{theorem}[THEOREM 12.4 in \cite{levin2017markov}]
    Let $W$ be the transition matrix of a reversible irreducible Markov chain with a finite state space $\Omega$ and a stationary distribution $\Pi$.
    Let $\Delta$ be the spectral gap of the chain.
    Let $\mathcal{P}$ be a set of all probability distributions on $\Omega$ and, for $n\in\mathbb{N}$ and $\mu\in\mathcal{P}$, denote by $\mu W^n$ the probability distribution after $n$ steps of the Markov chain with the initial distribution $\mu$.
    Then, for any $n\in\mathbb{N}$,
    \begin{equation}
        d(n)\le \frac{(1-\Delta)^n}{2\sqrt{\Pi_{\rm min}}},
    \end{equation}
    where $d(n):=\sup_{\mu\in\mathcal{P}}\|\mu W^n-\Pi\|_{\rm TV}$ and $\Pi_{\rm min}:=\min_{x\in\Omega}\Pi(x)$, and thus,    for any $\epsilon\in\mathbb{R}_+$,
    \begin{equation}
        t_{\rm mix}(\epsilon):=\min\{n\in\mathbb{N} \ | \ d(n)\le\epsilon \}\le \frac{1}{\Delta}\log\left(\frac{1}{\epsilon \Pi_{\rm min}}\right).
    \end{equation}
    \label{th:mixtime}
\end{theorem}

Since the early part of the Markov chain has not converged yet, we usually discard it, which is called {\it burn-in}.
Then, the procedure of the MH algorithm is summarized in Algorithm \ref{alg:MH}.\\

\begin{algorithm}[H]
    \begin{algorithmic}[1]
    \REQUIRE {
        \ \\
        \begin{itemize}
        \item Initial probability distribution $P_0$ on $\Omega$
        \item Burn-in length $n_{\rm b}\in\mathbb{N}$
        \item Number $n$ of sample values of $X$ we need 
        \item Function to sample from the proposal distribution $T$
        \item Function to compute $T$ 
        \item Function to compute the target probability $P$ except for the normalization constant
        \end{itemize}
    }

    \STATE Sample the initial point $x_0\in\Omega$ from $P_0$.

    \FOR{$i=0,...,n_{\rm b}+n-1$}
    
    \STATE Sample $\tilde{x}_{i+1}\in\Omega$ from the distribution $T(x_i,\cdot)$.

    \STATE Compute $A(x_i,\tilde{x}_{i+1})$ in Eq. (\ref{eq:accRatio}).

    \STATE Set $x_{i+1}=\tilde{x}_{i+1}$ with probability $A(x_i,\tilde{x}_{i+1})$ or $x_{i+1}=x_i$ otherwise.

    \ENDFOR

    \STATE Output $x_{n_b+1},...,x_{n_b+n}$.
    
    \caption{Metropolis-Hastings algorithm}\label{alg:MH}
    \end{algorithmic}
\end{algorithm}

The obtained sequence can be used for, e.g., calculating expectations of random variables. 
On the error in this, we have the following theorem.

\begin{theorem}[Theorem 11 in \cite{Rudolf+2010+323+342}]

Consider the Markov chain on a finite sample space $\Omega$ generated with a transition matrix $W$ and initial distribution $P_0$.
Denote its stationary distribution and spectral gap by $\Pi$ and $\Delta$, respectively.
Let the second largest eigenvalue of $W$ be $\lambda_1^\prime$ and let $\Delta^\prime:=1-\lambda_1^\prime$.
Then, for any function $f:\Omega\rightarrow \mathbb{R}$ and $n_{\rm b},n\in\mathbb{N}$,
\begin{align}
    e_{n,n_{\rm b},f}^2&:=\mathbb{E}_{\rm MC} \left[\left(S_{n,n_b,f}-\mathbb{E}_{\Pi}[f(x)]\right)^2\right] \nonumber \\
    &\le 
    \frac{2\|f\|_\infty^2}{n\Delta^\prime}+\frac{4\|\frac{P_0}{\Pi}-1\|_\infty^{1/2}\|f\|_\infty^2(1-\Delta)^{n_{\rm b}}}{n^2\Delta^2}
\end{align}
holds.
Here,
\begin{equation}
    S_{n,n_b,f} := \frac{1}{n}\sum_{i=1}^n f(x_{n_{\rm b}+i}),
\end{equation}
$x_i,i\in[n_{\rm b}+n]$ is the $i$th entry in the chain, $\mathbb{E}_{\rm MC}[\cdot]$ denotes the expectation with respect to the randomness of the generated chain, $\|f\|_\infty:=\max_{x\in\Omega}|f(x)|$, and $\|\frac{P_0}{\Pi}-1\|_\infty:=\max_{x\in\Omega}\left|\frac{P_0(x)}{\Pi(x)}-1\right|$.
\end{theorem}

$S_{n,n_b,f}$ is an estimate of $\mathbb{E}_{\Pi}[f(x)]$ the expectation of $f$ based on $n$ samples from the Markov chain with $n_{\rm b}$ burn-in samples discarded.
Then, the theorem implies that, to suppress the root mean square error $e_{n,n_{\rm b},f}$ of $S_{n,n_b,f}$ to $\epsilon\in(0,\|f\|_\infty)$, it is sufficient to take the burn-in length
\begin{equation}
n_{\rm b}=O\left(\frac{\log\left(\|\frac{P_0}{\Pi}-1\|_\infty^{1/2}\right)}{\Delta}\right)
\end{equation}
and the sample number
\begin{equation}
    n=O\left(\max\left\{\frac{\|f\|_\infty^2}{\Delta^\prime \epsilon^2},\frac{\|f\|_\infty}{\Delta \epsilon}\right\}\right). \label{eq:MCSampEps}
\end{equation}
Since $\Delta_1\ge\Delta$ holds, $n$ is also upper bounded as
\begin{equation}
    n=O\left(\frac{\|f\|_\infty^2}{\Delta \epsilon^2}\right). \label{eq:MCSampEps2}
\end{equation}
If $\Delta^\prime=\Delta$, which holds when $\lambda_1^\prime$ is the second largest eigenvalue also in modulus, that is, $\lambda_1=\lambda_1^\prime$, the bound (\ref{eq:MCSampEps}) becomes Eq. (\ref{eq:MCSampEps2}).
Including the burn-in, the total step number is also of order (\ref{eq:MCSampEps2}).

\subsubsection{Perturbed acceptance ratio \label{sec:pertAR}}

We consider the case where we can compute not the exact acceptance ratio but some approximation of it.
We expect that, in such a situation, although the chain converges to some distribution different from the original target distribution, the difference is small as far as the error in the acceptance ratio is small.
The bound on such a difference has been studied in previous studies \cite{alquier2016noisy,medina2016stability,Rudolf2018}, and, in this paper, we use the result of \cite{alquier2016noisy}.
As a preparation to present it, we introduce and a concept called uniform ergodicity \cite{Roberts2004}.

\begin{definition}
    If, for a Markov chain on a finite state space, there exist $\rho\in(0,1)$ and $C\in\mathbb{R}$ such that
    \begin{equation}
        \|\Pi_n - \Pi\|\le C\rho^n
    \end{equation}
    holds for any initial distribution and $n\in\mathbb{N}$, where $\Pi_n$ is the distribution after $n$ steps and $\Pi$ is the stationary distribution, we say that the Markov chain is $(C,\rho)$-uniformly ergodic.
\end{definition}

Apparently, for a reversible irreducible Markov chain, an upper bound of $\rho$ is $1-\Delta$.


Then, the theorem we will later use is the following.
This is the restriction of Corollary 2.3 in \cite{alquier2016noisy} to the case that we can compute an approximate acceptance ratio that is deterministic and has a bounded error.

\begin{theorem}

Consider Algorithm \ref{alg:MH}. 
Suppose that the generated Markov chain $\mathcal{C}$ is $(C,\rho)$-uniformly ergodic.
Also consider another chain $\mathcal{C}^\prime$ same as $\mathcal{C}$ expect that the acceptance ratio $A$ is replaced with $A^\prime:\Omega\times\Omega \rightarrow [0,1]$ such that
\begin{equation}
    \max_{x,y\in\Omega} |A^\prime(x,y)-A(x,y)| \le \epsilon
\end{equation}
with some $\epsilon\in \mathbb{R}$. 
Then, the stationary distributions $\Pi$ and $\Pi^\prime$ of $\mathcal{C}$ and $\mathcal{C}^\prime$ satisfy
\begin{equation}
    \|\Pi-\Pi^\prime\|_{\rm TV} \le \epsilon\left(\lambda+\frac{C\rho^\lambda}{1-\rho}\right) \label{eq:AppARDistDiff}
\end{equation}
with $\lambda:=\left\lceil\frac{\log(1/C)}{\log\rho}\right\rceil$.

\end{theorem}

Because of Theorem \ref{th:mixtime}, for a reversible irreducible Markov chain, we obtain
\begin{equation}
    \|\Pi-\Pi^\prime\|_{\rm TV} \le \epsilon\left(\left\lceil\frac{\log(2\sqrt{\Pi_{\rm min}})}{\log(1-\Delta)}\right\rceil+\frac{1}{\Delta}\right), \label{eq:AppARDistDiffMH}
\end{equation}
by substituting $\rho=1-\Delta$ and $C=\frac{1}{2\sqrt{\Pi_{\rm min}}}$ in Eq. (\ref{eq:AppARDistDiff}).

\subsection{Quantum walk operator}

Among several versions of quantum walk operators proposed so far, this paper adopts that in \cite{Lemieux2020efficientquantum}.
It is dedicated for the MH algorithm for Ising models, and will be generalized in Sec. \ref{sec:ourWalkOp}.
Suppose that we are now considering an Ising system with $n_{\rm sp}$ spins, which means that $\Omega=\{-1,1\}^{\times n_{\rm sp}}$.
We use a system of two quantum registers $R_{\rm S}$ and $R_{\rm M}$ and a qubit $R_{\rm C}$.
$R_{\rm S}$ hold a string $x$ of $\pm 1$ with length $n_{\rm sp}$, that represents the current spin configuration.
$R_{\rm M}$ holds a bit string $z\in\{0,1\}^{\times n_{\rm sp}}$ that represents the next spin flip: if the $i$th bit $z_i$ of $z$ is 1 (resp. 0), the $i$th spin $x_i$ in $x\in\{-1,1\}^{\times n_{\rm sp}}$ is changed to $-x_i$ (resp. unchanged).
We denote by $x \odot z$ the spin configuration generated by $x$ and $z$ under this rule.
Then, we consider the following operator
\begin{equation}
    U_{\rm IS}=R_{\rm IS}V^\dagger_{\rm IS} B^\dagger_{\rm IS} F_{\rm IS}B_{\rm IS}V_{\rm IS}. \label{eq:walkOp}
\end{equation}
The component operators are as follow.
$V_{\rm IS}$ is a unitary on $R_{\rm M}$ that acts as
\begin{equation}
    V_{\rm IS}\ket{0}_{R_{\rm M}}=\sum_{z\in\mathcal{M}} \sqrt{p_{\rm fl}(z)}\ket{z}_{R_{\rm M}}.
\end{equation}
Here, $\mathcal{M}\subset \{0,1\}^{\times n_{\rm sp}}$ is the set of possible spin flips and $p_{\rm fl}:\mathcal{M}\rightarrow (0,1)$ is a probability distribution on $\mathcal{M}$, which we use as the proposal distribution of the next flip.
Note that it is assumed that possible flips $\mathcal{M}$ and the probability $p_{\rm fl}(z)$ that a flip $z\in\mathcal{M}$ is proposed are independent of the current spin configuration.
We can associate $p_{\rm fl}$ with the proposal distribution $T$ in Algorithm \ref{alg:MH} as
\begin{equation}
    T(x,y)=
    \begin{cases}
        p_{\rm fl}(z) & ; \ {\rm if} \ y=x\odot z \ {\rm with \ some} \ z\in\mathcal{M} \\
        0 & ; \ {\rm otherwise}
    \end{cases}, \label{eq:TIsing}
\end{equation}
where $x,y\in\{0,1\}^{\times n_{\rm sp}}$.
$B_{\rm IS}$ is the rotation gate on $R_{\rm C}$ controlled by $R_{\rm S}$ and $R_{\rm M}$, whose rotation angle is determined by the acceptance ratio.
That is, it acts as
\begin{align}
    & B_{\rm IS}\ket{x}_{R_{\rm S}}\ket{z}_{R_{\rm M}}\ket{\phi}_{R_{\rm C}}= \nonumber\\
    & \ \ \ket{x}_{R_{\rm S}}\ket{z}_{R_{\rm M}}\otimes
    \begin{pmatrix}
    \sqrt{1-A(x,x\odot z)} & -\sqrt{A(x,x\odot z)} \\
    \sqrt{A(x,x\odot z)} & \sqrt{1-A(x,x\odot z)}
    \end{pmatrix}
    \ket{\phi}_{R_{\rm C}}
\end{align}
for any $x\in\Omega$, $z \in \mathcal{M}$ and state $\ket{\phi}_{R_{\rm C}}$ on $R_{\rm C}$, where $A(x,x\odot z)=P(x\odot z)T(x\odot z,x)/P(x)T(x,x\odot z)$ with some target distribution $P$ on $x\in\{-1,1\}^{\times n_{\rm sp}}$.
$F_{\rm IS}$ is a gate to apply the spin flip under the control by $R_{\rm C}$, which acts as
\begin{align}
    &F_{\rm IS}\ket{x}_{R_{\rm S}}\ket{z}_{R_{\rm M}}\ket{\phi}_{R_{\rm C}}= \nonumber \\
    & \quad
    \begin{cases}
    \ket{x}_{R_{\rm S}}\ket{z}_{R_{\rm M}}\ket{0}_{R_{\rm C}} & {\rm if} \ \ket{\phi}_{R_{\rm C}}=\ket{0}_{R_{\rm C}} \\
    \ket{x \odot z}_{R_{\rm S}}\ket{z}_{R_{\rm M}}\ket{1}_{R_{\rm C}} & {\rm if} \ \ket{\phi}_{R_{\rm C}}=\ket{1}_{R_{\rm C}}
    \end{cases}.
\end{align}
Finally,
\begin{equation}
    R_{\rm IS}=2\Lambda_{0,{\rm IS}}-I_{R_{\rm S}} \otimes I_{R_{\rm M}} \otimes I_{R_{\rm C}}, \label{eq:R}
\end{equation}
where
\begin{equation}
    \Lambda_{0,{\rm IS}}:=I_{R_{\rm S}}\otimes\ket{0}\bra{0}_{R_{\rm M}}\otimes\ket{0}\bra{0}_{R_{\rm C}}. \label{eq:Pi0}
\end{equation}

Although this is different from quantum walk operators in previous studies \cite{Szegedy2004,Somma0712.1008,Somma2008,Wocjan2008,Magniez2011,Yung2012,Harrow2020}, it has the following property on its spectrum, which is same as previous ones, and thus can be used as an alternative.

\begin{theorem}
    Consider the Markov chain generated by Algorithm \ref{alg:MH} with the state space $\Omega=\{-1,1\}^{\times n_{\rm sp}}$, the target distribution $P$ and the proposal distribution $T$ in Eq. (\ref{eq:TIsing}).
    Denote by $\Delta$ its spectral gap.
    Let $\mathcal{A}={\rm span}\{\ket{x}_{R_{\rm S}}\ket{0}_{R_{\rm M}}\ket{0}_{R_{\rm C}} | x\in\Omega\}$ and $\mathcal{B}=V^\dagger_{\rm IS} B^\dagger_{\rm IS} F_{\rm IS}B_{\rm IS}V_{\rm IS}\mathcal{A}$.
    Then, on $\mathcal{A}+\mathcal{B}$, $\ket{P}$ is the unique eigenstate of $U_{\rm IS}$ with eigenvalue 1, and any other eigenvalue is written as $e^{i\theta}$ with $\theta\in\mathbb{R}$ such that $|\theta|\ge\arccos(1-\Delta)$.
    \label{th:phasegap}
\end{theorem}

We call operators that have this property the quantum walk operators for the Markov chain.
Although this theorem only states on the phase gap of the quantum walk operator unlike the previous results such as Theorem 1 in \cite{Szegedy2004} and Theorem 1 in \cite{Wocjan2008}, which state on the entire spectrum in more detail, it is sufficient for our purpose for the reason explained in Sec. \ref{sec:QSA}.

\subsection{Quantum simulated annealing \label{sec:QSA}}

In previous studies on quantum versions of MCMC \cite{Somma0712.1008,Somma2008,Wocjan2008,Yung2012,Harrow2020}, the aim is generating a quantum state in which the target distribution $P$ is encoded in the amplitude, that is,
\begin{equation}
\ket{P}:=\sum_{x\in\Omega}\sqrt{P(x)}\ket{x}.
\end{equation}
In this paper, we use the method proposed in \cite{Harrow2020}.
In that paper, following \cite{Wocjan2008}, the author took the strategy called QSA, which was inspired by simulated annealing.
That is, assuming that $P$ is in the form of
\begin{equation}
P(x)\propto P_0(x)\exp(- L(x))
\end{equation}
with the {\it prior distribution} $P_0$ and the {\it negative log-likelihood} $L(x)$, we sequentially prepare the quantum states $\ket{P_{\beta_1}},...,\ket{P_{\beta_l}}$ from the initial state $\ket{P_0}$.
Here, these states encode the distributions in the form of
\begin{equation}
P_\beta(x)\propto P_0(x)\exp(-\beta L(x)) \label{eq:Pbeta}
\end{equation}
with parameters called temperatures $\beta_0=0<\beta_1<\cdots<\beta_l=1$, and largely overlap: $|\braket{P_{\beta_i} | P_{\beta_{i+1}}}|^2\ge p$ with $p=\Theta(1)$.
The method in \cite{Harrow2020} is twofold: obtain a set of appropriate values of $\{\beta_i\}$ and transform $\ket{P_0}$ to $\ket{P_{\beta_l}}$ via $\ket{P_{\beta_1}},...,\ket{P_{\beta_{l-1}}}$.
Both of two phases are based on the following result, how to approximately construct the phase gate, which multiplies a phase factor $\omega$ to the state vector for the state $\ket{P_{\beta_i}}$ but not for orthogonal states, using the quantum walk operator for the Markov chain converging to $P_{\beta_i}$.
This is summarized as the following theorem.

\begin{theorem}[Corollary 2 in \cite{Wocjan2008}]

Consider a Markov chain $\mathcal{C}$ on a finite state space $\Omega$ with the transition matrix $W$, the stationary distribution $P$ and the spectral gap $\Delta$.
Let $\omega$ be a complex number with unit modulus and define a unitary
\begin{equation}
    R^{\omega}_{\ket{P}}:=\omega \Lambda^\parallel_{\ket{P}}+\Lambda^\perp_{\ket{P}}, \label{eq:RomegaP}
\end{equation}
where $\Lambda^\parallel_{\ket{P}}$ and $\Lambda^\perp_{\ket{P}}$ are the projection onto the subspace spanned by a state $\ket{P}$ on a register $R$ and that onto the orthogonal subspace, respectively, 
Then, for any $\delta\in(0,1)$, we have an access to a unitary operator $\tilde{R}^\omega_{\ket{P},\delta}$ that has the following properties:
\begin{itemize}
    \item acts on a system of $R$ and $n_{\rm anc}$ ancillary qubits, where $n_{\rm anc}=O\left(\log\left(\frac{1}{\Delta}\right)\log\left(\frac{1}{\delta}\right)\right)$
    
    \item uses the controlled version of the quantum walk operator $U$ for $\mathcal{C}$ and its inverse $O\left(\frac{\log\left(1/\delta\right)}{\sqrt{\Delta}}\right)$ times
    
    \item for any state $\ket{\Xi}$ on $R$, $\tilde{R}^{\omega}_{\ket{P},\delta}\ket{\Xi}\ket{0}^{\otimes n_{\rm anc}}=(R^{\omega}_{\ket{P}}\ket{\Xi})\ket{0}^{\otimes n_{\rm anc}}+\ket{\xi}$, where $\ket{\xi}$ is an unnormalized state on the entire system with $\|\ket{\xi}\|\le \delta$.

\end{itemize}
   \label{th:Romega}
\end{theorem}

Note that, as special cases, this theorem covers the original version of Corollary 2 in \cite{Wocjan2008} with $\omega=\omega_{\pi/3}:=e^{i\frac{\pi}{3}}$, and Theorem 6 in \cite{Magniez2011} with $\omega=-1$, which corresponds to the reflection operator with respect to $\ket{P}$.

The outline of constructing $R^{\omega}_{\ket{P}}$ with $U$ is as follows.
Given $\ket{\psi_i}$ an eigenstate of $U$ with the corresponding eigenvalue $\lambda_i=e^{i\theta_i}$, we can compute an estimate of the phase $\theta_i$ onto another register $R_{\rm ph}$ by quantum phase estimation (QPE) \cite{kitaev1995quantum,Cleve1998}.
Because of Theorem \ref{th:phasegap}, the difference between the phase of the eigenstate $\ket{P}$, which is 0, and that of any other eigenstate is larger than $\arccos(1-\Delta)=\Omega(\sqrt{\Delta})$, and thus $\ket{P}$ can be distinguished from other eigenstates via QPE with accuracy $O(\sqrt{\Delta})$.
Thus, the above QPE followed by a phase gate controlled by the register $R_{\rm ph}$ is the phase gate that acts only on $\ket{P}$, that is, $R^{\omega}_{\ket{P}}$.
Because QPE outputs an estimate within the desired accuracy not certainly but with a finite failure probability, this implementation of $R^{\omega}_{\ket{P}}$ gives an approximate gate $\tilde{R}^\omega_{\ket{P},\delta}$.

Given the phase gate with $\omega=\omega_{\pi/3}$, we can generate $\ket{P_{\beta_{i+1}}}$ from $\ket{P_{\beta_i}}$ by Grover's $\frac{\pi}{3}$-amplitude amplification \cite{Grover2005}, which is summarized as follows.

\begin{theorem}
    Let $\ket{\phi_1}$ and $\ket{\phi_2}$ be quantum states on a same register satisfying $|\braket{\phi_1 | \phi_2}|^2 \ge p$ with some $p\in(0,1]$.
    For $i\in\{1,2\}$, denote by $\Lambda^\parallel_{\ket{\phi_i}}$ and $\Lambda^\perp_{\ket{\phi_i}}$ the projection on the subspace spanned by $\ket{\phi_i}$ and that on the orthogonal subspace, respectively, and define the unitary $R^{\omega_{\pi/3}}_{\ket{\phi_i}}:=\omega_{\pi/3}\Lambda^\parallel_{\ket{\phi_i}}+\Lambda^\perp_{\ket{\phi_i}}$.
    Define the unitaries $U_{i,m}$ recursively as follows:
    \begin{eqnarray}
        U_{i,0}&=&I \nonumber \\
        U_{i,m+1}&=& U_{i,m} R^{\omega_{\pi/3}}_{\ket{\phi_i}} U_{i,m}^\dagger R^{\omega_{\pi/3}}_{\ket{\phi_{i+1}}} U_{i,m}.
    \end{eqnarray}
    Then, for any $m\in\mathbb{N}$,
    \begin{equation}
        |\bra{\phi_2} U_{i,m} \ket{\phi_1}|^2 \ge 1-(1-p)^{3^m}
    \end{equation}
    holds.
    This implies that, for any $\epsilon\in\mathbb{R}_+$, we can prepare $\widetilde{\ket{{\phi}_2}}$ $\epsilon$-close to $\ket{\phi_2}$ from $\ket{\phi_1}$ using the unitaries in $\left\{R^{\omega_{\pi/3}}_{\ket{\phi_1}},R^{\omega_{\pi/3}}_{\ket{\phi_2}},\left(R^{\omega_{\pi/3}}_{\ket{\phi_1}}\right)^\dagger,\left(R^{\omega_{\pi/3}}_{\ket{\phi_2}}\right)^\dagger\right\}$ $O\left(\frac{\log\left(\frac{1}{\epsilon}\right)}{\log\frac{1}{1-p}}\right)$ times.
\end{theorem}

Strictly speaking, we do not have the exact operator $R^{\omega_{\pi/3}}_{\ket{P}}$ but an approximate one $\tilde{R}^{\omega_{\pi/3}}_{\ket{P},\epsilon}$.
According to \cite{Wocjan2008}, we can transform $\ket{P_0}$ to a state $\widetilde{\ket{P_1}}$ $\epsilon$-close to $\ket{P_1}$, via $\widetilde{\ket{P_{\beta_1}}},...,\widetilde{\ket{P_{\beta_{l-1}}}}$ the approximate states of $\ket{P_{\beta_1}},...,\ket{P_{\beta_{l-1}}}$, by $\frac{\pi}{3}$-amplitude amplifications with $\tilde{R}^{\omega_{\pi/3}}_{\ket{P_{\beta_i}},\epsilon^\prime}$ used instead of $R^{\omega_{\pi/3}}_{\ket{P_{\beta_i}}}$, where $\epsilon^\prime$ is some real number set depending on $l$, $\epsilon$ and $p$.
We hereafter call this method the approximate $\frac{\pi}{3}$-amplitude amplification (A$\frac{\pi}{3}$AA) with accuracy $\epsilon$ and overlap $p$.
Its complexity is summarized as follows.

\begin{theorem}[Theorem 5 in \cite{Harrow2020}, originally Theorem 2 in \cite{Wocjan2008}]

Consider $l$ Markov chains $\mathcal{C}_1,...,\mathcal{C}_l$ on a finite state space $\Omega$ with stationary distributions $p_1,...,p_l$ and spectral gaps lower bounded by $\Delta\in(0,1)$.
Let $p_0$ be another probability distribution on $\Omega$ and suppose that the state $\ket{p_0}$ is given on a register $R$.
Assume that, for some $p\in(0,1)$, $|\braket{p_i | p_{i+1}}|^2\ge p$ holds for any $i\in[l]_0$.
Then, for any $\epsilon\in(0,1)$, we have an access to a unitary operator $U_{\rm QSA}$ on the system of $R$ and $n_{\rm anc}$ qubits that acts as
\begin{equation}
    U_{\rm QSA}\ket{0}_R\ket{0}^{\otimes n_{\rm anc}} = \ket{p_l}_R\ket{0}^{\otimes n_{\rm anc}} + \ket{\zeta}
\end{equation}
making $O\left(\frac{l}{\sqrt{\Delta}}\log^2\frac{l}{p\epsilon}\log\frac{1}{p}\right)$ queries to the controlled quantum walk operators for $\mathcal{C}_1,...,\mathcal{C}_l$.
Here,
\begin{equation}
n_{\rm anc}=O\left(\log\left(\frac{1}{\Delta}\right)\log\left(\frac{l\log\left(\frac{l}{\epsilon}\right)}{\log\left(\frac{1}{1-p}\right)}\right)\right),
\end{equation}
and $\ket{\zeta}$ is an (unnormalized) state on the entire system with $\|\ket{\zeta}\|\le\epsilon$. 
    
\end{theorem}

Besides, given the phase gate with $\omega=-1$, which is namely the reflection operator, we can use nondestructive amplitude estimation (NAE) \cite{Harrow2020}, a modification of QAE \cite{brassard2002quantum}, to estimate $|\braket{P_{\beta_{i}} | P_{\beta_{i+1}}}|^2$.

\begin{theorem}[Theorem 6 in \cite{Harrow2020}]
Given a quantum state $\ket{\phi}$ on a register $R$ and two operators 
$R_{\phi}=2\ket{\phi}\bra{\phi}-I$ and $R_{\phi^\prime}=2\ket{\phi^\prime}\bra{\phi^\prime}-I$, where $\ket{\phi^\prime}$ is another state on $R$, for any $\epsilon,\delta\in(0,1)$, there exists an quantum algorithm with following properties:
\begin{itemize}
    \item with probability at least $1-\delta$, outputs an $\epsilon$-approximation of $|\braket{\phi^\prime | \phi}|^2$ and a flag 1, and restores the state $\ket{\phi}$
    \item otherwise, output a flag 0
    \item uses $R_{\phi}$ and $R_{\phi^\prime}$ $O\left(\frac{\log\left(1/\delta\right)}{\epsilon}\right)$ times.
\end{itemize}
\end{theorem}

Again, we can use only approximations of reflection operators.
\cite{Harrow2020} showed that, with probability at least $1-\delta$, NAE using approximate reflection operators instead of exact ones output an $\epsilon$-approximation of $|\braket{\Pi^\prime | \Pi}|^2$ for stationary distributions $\Pi$ and $\Pi^\prime$ of some Markov chains. 
We hereafter call this approximate NAE (ANAE) with accuracy $\epsilon$ and failure probability $\delta$.
The following theorem states on its complexity.

\begin{theorem}[Theorem 9 in \cite{Harrow2020}]
    Consider Markov chains $\mathcal{C}_1$ and $\mathcal{C}_2$ on a finite state space $\Omega$ with stationary distributions $\Pi_1$ and $\Pi_2$ and spectral gaps lower bounded by $\Delta\in(0,1)$.
    Suppose that the state $\ket{\Pi_1}$ is given.
    Then, there is a quantum algorithm with the following properties:
    \begin{itemize}
    \item with probability at least $1-\delta$, outputs an $\epsilon$-approximation of $|\braket{\Pi_1 | \Pi_2}|^2$ and a flag 1, and restores the state $\ket{\Pi_1}$
    \item otherwise, output a flag 0
    \item uses the controlled quantum walk operators for $\mathcal{C}_1$ and $\mathcal{C}_2$ $O\left(\frac{1}{\epsilon\sqrt{\Delta}}\log\left(\frac{1}{\epsilon}\right)\log\left(\frac{1}{\delta}\right)\right)$ times.
    \end{itemize}
\end{theorem}

Then, combining the above building blocks yields the method proposed in \cite{Harrow2020}.
It is summarized as Algorithm \ref{alg:QSA}.

\begin{algorithm}[H]
    \begin{algorithmic}[1]
    \REQUIRE {
        \ \\
        \begin{itemize}        
        \item Access to a unitary operator $O_{P_0}$ to generate $\ket{P_0}$
        \begin{equation}
            O_{P_0}\ket{0}=\ket{P_0} \label{eq:OP0}
        \end{equation}

        \item For any $\beta\in(0,1]$, an access to a quantum walk operator $U_\beta$ for a Markov chain that has $P_\beta$ in Eq. (\ref{eq:Pbeta}) as the stationary distribution and the spectral gap lower bounded $\Delta_{\rm min}$.



        \item Accuracy $\epsilon\in\mathbb{R}_+$ for the final state

        \item Failure probability $\eta\in(0,1)$
        
        \end{itemize}
    }

    \ENSURE{
        Either of 
        \begin{enumerate}
            \renewcommand{\labelenumi}{(\Alph{enumi})}
            \item sequence $\beta_0=0<\beta_1<\cdots<\beta_{l-1}<\beta_l=1$ such that $l\le l_{\rm max}$ and $|\braket{P_{\beta_i} | P_{\beta_{i+1}}}|^2\ge \frac{9}{10e^2}$ for any $i\in[l]_0$, and ${\rm flg}=1$
            \item ${\rm flg}=0$
        \end{enumerate}
    }

    \STATE Set $\widetilde{\ket{P_0}}=\ket{P_0}$.

    \FOR{$i=0, 1, ..., l_{\rm max}-1$}
    
    \STATE Find the largest $\beta^\prime\in(\beta_i,1]$ such that $|\braket{P_{\beta_i} | P_{\beta^\prime}}|^2\ge e^{-2}$ by binary search with precision $1/L_{\rm max}$.
    Here, $|\braket{P_{\beta_i} | P_{\beta^\prime}}|^2$ is computed by ANAE with accuracy $1/10e^2$ and failure probability $\eta/l_{\rm max}L_{\rm max}$, with $\widetilde{\ket{P_{\beta_i}}}$ used instead of $\ket{P_{\beta_i}}$.

    \IF{at least one ANAE in line 3 returns a flag 0}
    \STATE {Output ${\rm flg}=0$ and stop.}
    \ENDIF
    
    \STATE Let the result in line 3 be $\beta_{i+1}$.

    \IF{$\beta_{i+1}=1$}

    \STATE Output $\beta_0=0,\beta_1,...,\beta_{i+1}$ and ${\rm flg}=1$, and stop.

    \ENDIF

    \STATE Generate $\widetilde{\ket{P_{\beta_{i+1}}}}$ from $\widetilde{\ket{P_{\beta_i}}}$ by A$\frac{\pi}{3}$AA with accuracy $\frac{\epsilon}{l_{\rm max}}$ and overlap $\frac{9}{10e^2}$.

    \ENDFOR

    \STATE Output ${\rm flg}=0$.
    
    \caption{Quantum simulated annealing (Algorithm 1 in \cite{Harrow2020}, modified)}\label{alg:QSA}
    \end{algorithmic}
\end{algorithm}

Here, $l_{\rm max}:=\sqrt{\bar{L}\log\bar{L}}$, where $\bar{L}:=\mathbb{E}_{P_0}[L(x)]$, and $L_{\rm max}:=\max_{x\in\Omega} L(x)$.

The complexity of this algorithm is stated in Theorem \ref{th:QSA}.

\begin{theorem}[Theorem 10 in \cite{Harrow2020}]
    Algorithm \ref{alg:QSA} yields output (A) with probability at least $1-\eta$, calling operators in $\{U_\beta \ | \ \beta\in(0,1]\}$
    \begin{equation}
        O\left(\frac{l_{\rm max}}{\sqrt{\Delta_{\rm min}}}\left(\log^2 l_{\rm max}+\log  L_{\rm max}\log \left(\frac{l_{\rm max}L_{\rm max}}{\eta}\right)\right) \right)
    \end{equation}
    times.
    For the obtained $\beta_0,\beta_1,...,\beta_l$, using A$\frac{\pi}{3}$AA with accuracy $\epsilon$ and overlap $\frac{9}{10e^2}$ for Markov chains with stationary distributions $P_{\beta_0},P_{\beta_1},...,P_{\beta_l}$, we can generate the state
    \begin{equation}
    \widetilde{\ket{P_{\beta_l}}}:=\ket{P_{\beta_l}}\ket{0}^{\otimes n_{\rm anc}} + \ket{\xi},
    \end{equation}
    where $n_{\rm anc}=O\left(\log\left(\frac{1}{\Delta_{\rm min}}\right)\log\left(\frac{l_{\rm max}}{\epsilon}\right)\right)$ and $\ket{\xi}$ is an unnormalized state such that $\|\ket{\xi}\|\le\epsilon$.
    In this process, operators $U_{\beta_1},...,U_{\beta_l}$ are called
    \begin{equation}
        O\left(\frac{l_{\rm max}}{\sqrt{\Delta_{\rm min}}}\log^2\left(\frac{l_{\rm max}}{\epsilon}\right) \right) \label{eq:compQSA}
    \end{equation}
    times.
    \label{th:QSA}
\end{theorem}

\subsection{Quantum Monte Carlo integration}

\cite{Montanaro2015} presented a quantum algorithm to calculate an expected value of a random variable, which we call QMCI in this paper.

\begin{theorem}[Theorem 2.3 in \cite{Montanaro2015}]
    Let $P$ be a probability distribution on a finite sample space $\Omega\subset\mathbb{R}^d$.
    Suppose that we have a quantum circuit $O_{P}$ on a two-register system that acts as $O_{P}\ket{0}\ket{0}=\sum_{x\in\Omega} \sqrt{P(x)}\ket{\phi_x}\ket{x}$, where $\ket{\phi_x}$ is some state on the first register.
    Also suppose that, for a function $F:\Omega\rightarrow[0,1]$, we have a quantum circuit $O_{F}$ on a two-register system that acts as $O_{F}\ket{x}\ket{0}=\ket{x}\ket{F(x)}$ for any $x\in\Omega$.
    Then, for any $\epsilon\in\mathbb{R}_+$ and $\delta\in(0,1)$, there is a quantum algorithm that, with probability at least $1-\delta$, outputs an $\epsilon$-approximation of $\mu_F:=\sum_{x\in\Omega}P(x)F(x)$, making
    \begin{equation}
        O\left(\frac{1}{\epsilon}\log\delta^{-1}\right) \label{eq:compQMCIBound}
    \end{equation}
    uses of $O_{P}$ and $O_F$.
    \label{th:QMCIBound}
\end{theorem}

The above theorem is on a version of the algorithm for a bounded integrand $F$.
\cite{Montanaro2015} presented another version for an integrand with a bounded variance.
We now present a modification of this so that it can be used in QSA.
Namely, we aim to obtain not an approximate value of an expectation but a quantum state in which approximations are encoded, making no measurement.
Besides, we concentrate on the situation that we compute the mean of a finite number of real numbers, which is sufficient for our purpose.

\begin{theorem}
    Let $M$ be a positive integer and $\mathcal{X}$ be a set of $M$ real numbers, $X_0,...,X_{M-1}$, whose mean is $\mu:=\frac{1}{M}\sum_{i=0}^{M-1}X_i$ and sample variance satisfies $\frac{1}{M}\sum_{i=0}^{M-1}X_i^2 - \mu^2 \le \sigma^2$ with some $\sigma\in\mathbb{R}_+$.
    Suppose that we are given an access to a unitary operator $O_\mathcal{X}$ that acts as
    \begin{equation}
        O_{\mathcal{X}}\ket{i}\ket{0}=\ket{i}\ket{X_i}, \label{eq:OX}
    \end{equation}
    for any $i\in[M]_0$.
    Then, for any 
    $\epsilon\in\mathbb{R}_+$ and $\delta\in(0,1)$, we have an access to a unitary operator $O_{\mathcal{X},\epsilon,\delta,\sigma}^{\rm mean}$ that acts on a system of two registers $R_1$ and $R_2$ as
    \begin{equation}
        O_{\mathcal{X},\epsilon,\delta,\sigma}^{\rm mean}\ket{0}_{R_1}\ket{0}_{R_2}= \ket{0}_{R_1}\ket{\tilde{\mu}}_{R_2} + \gamma\ket{\psi}_{R_1,R_2}, \label{eq:OmeanX}
    \end{equation}
    where $\tilde{\mu}$ is an $\epsilon$-approximation of $\mu$, $\ket{\psi}_{R_1,R_2}$ is a state on the entire system,
    and $\gamma\in\mathbb{C}$ satisfies $|\gamma|^2\le\delta$.
    In $O_{\mathcal{X},\epsilon,\delta,\sigma}^{\rm mean}$, $O_\mathcal{X}$ is used
    \begin{equation}
        O\left(\frac{\sigma}{\epsilon}\log^{3/2}\left(\frac{\sigma}{\epsilon}\right)\log\log\left(\frac{\sigma}{\epsilon}\right)\log\left(\frac{1}{\delta}\right)\right) \label{eq:compQMCI}
    \end{equation}
    time.
    The total qubit number in the system of $R_1$ and $R_2$ is of order
    \begin{equation}
        O\left(\left(\log M+\log\left(\frac{\sigma}{\epsilon}\right)\right)\log\left(\frac{\sigma}{\epsilon}\right)\log\log\left(\frac{\sigma}{\epsilon}\right)\log\delta^{-1}\right).
        \label{eq:qubitQMCI}
    \end{equation} 
    \label{th:QMCI}
\end{theorem}

Although this theorem resembles Theorem 5 in \cite{Miyamoto2022}, there is a following difference.
$O_{\mathcal{X},\epsilon,\delta,\sigma}^{\rm mean}$ in \cite{Miyamoto2022}, which we rename $\tilde{O}_{\mathcal{X},\epsilon,\delta,\sigma}^{\rm mean}$, generates a superposition of $\ket{y_1},\ket{y_2},...$, where $\{y_i\}$ are real numbers close to $\mu$.
On the other hand, the state in Eq. (\ref{eq:OmeanX}) is almost equal to a product state of $\ket{\tilde{\mu}}$, a computational basis state corresponding to one approximation of $\mu$, and $\ket{0}$, except a small residual term $\gamma\ket{\psi}$.
This is realized by combining $\tilde{O}_{\mathcal{X},\epsilon,\delta,\sigma}^{\rm mean}$ and rounding.
Including this point, the proof of Theorem \ref{th:QMCI} is presented in Appendix \ref{sec:QMCIDetail}.

\section{Proposed algorithm \label{sec:ourAlg}}

Now, let us present the proposed algorithm, the quantum MH algorithm with the target distribution estimated by QMCI.

\subsection{Modified quantum walk operator \label{sec:ourWalkOp}}

We start from generalizing the quantum walk operator in Eq. (\ref{eq:walkOp}) for Ising models to that for the Markov chain generated by Algorithm \ref{alg:MH} with a general finite state space $\Omega\subset\mathbb{R}^d$.
We define
\begin{equation}
    U=RV^\dagger B^\dagger SFBV.  \label{eq:walkOpOurs}
\end{equation}
This acts on a system of two quantum registers $R_{\rm S}$ and $R_{\rm M}$, which now have a sufficient number of qubits to represent real vectors, and a qubit $R_{\rm C}$.
$V$ acts on the system of $R_{\rm S}$ and $R_{\rm M}$ as
\begin{equation}
    V\ket{x}_{R_{\rm S}}\ket{0}_{R_{\rm M}}=\ket{x}_{R_{\rm S}}\sum_{\Delta x\in\Omega_x} \sqrt{T(x,x+\Delta x)}\ket{\Delta x}_{R_{\rm M}} \label{eq:V}
\end{equation}
for any $x\in\Omega$, where
\begin{equation}
    \Delta\Omega_x := \left\{\Delta x \in \mathbb{R}^d \  \middle| \ x+\Delta x\in\Omega \right\}
\end{equation}
is the set of possible moves in a transition from $x$.
$B$
acts as
\begin{align}
    & B\ket{x}_{R_{\rm S}}\ket{\Delta x}_{R_{\rm M}}\ket{\phi}_{R_{\rm C}}= \nonumber\\
    & \ \ \ket{x}_{R_{\rm S}}\ket{\Delta x}_{R_{\rm M}} \nonumber \\
    & \ \ \otimes
    \begin{pmatrix}
    \sqrt{1-A(x,x+\Delta x)} & -\sqrt{A(x,x+\Delta x)} \\
    \sqrt{A(x,x+\Delta x)} & \sqrt{1-A(x,x+\Delta x)}
    \end{pmatrix}
    \ket{\phi}_{R_{\rm C}}
\end{align}
for any $x\in\Omega$, $\Delta x \in \Omega_x$ and state $\ket{\phi}$ on $R_{\rm C}$.
$F$ makes a state transition, which is implemented by an adder gate controlled by $R_{\rm C}$, that is,
\begin{align}
    &F\ket{x}_{R_{\rm S}}\ket{\Delta x}_{R_{\rm M}}\ket{\phi}_{R_{\rm C}}= \nonumber \\
    & \quad
    \begin{cases}
    \ket{x}_{R_{\rm S}}\ket{\Delta x}_{R_{\rm M}}\ket{0}_{R_{\rm C}} & {\rm if} \ \ket{\phi}=\ket{0}_{R_{\rm C}} \\
    \ket{x+\Delta x}_{R_{\rm S}}\ket{\Delta x}_{R_{\rm M}}\ket{1}_{R_{\rm C}} & {\rm if} \ \ket{\phi}_{R_{\rm C}}=\ket{1}_{R_{\rm C}}
    \end{cases}.
\end{align}
The unitary $S$, for which Eq. (\ref{eq:walkOp}) has no counterpart, acts on the system of $R_{\rm M}$ and $R_{\rm C}$ to flip the sign of the value on $R_{\rm M}$ under the control by $R_{\rm C}$:
\begin{equation}
    S\ket{\Delta x}_{R_{\rm M}}\ket{\phi}_{R_{\rm C}}=
    \begin{cases}
    \ket{\Delta x}_{R_{\rm M}}\ket{0}_{R_{\rm C}} & {\rm if} \ \ket{\phi}_{R_{\rm C}}=\ket{0}_{R_{\rm C}} \\
    \ket{-\Delta x}_{R_{\rm M}}\ket{1}_{R_{\rm C}} & {\rm if} \ \ket{\phi}_{R_{\rm C}}=\ket{1}_{R_{\rm C}}
    \end{cases}.
\end{equation}
In other words, $S$ converts the move from $x$ to $y$ to the inverse move from $y$ to $x$.
We can consider that an identity operator is contained in Eq. (\ref{eq:walkOp}) as a counterpart for $S$, since any spin flip is the inverse transform of itself.
Finally, $R$ is same as $R_{\rm IS}$ in Eq. (\ref{eq:R}).
$U$ in Eq. (\ref{eq:walkOpOurs}) also has the property same as $U_{\rm IS}$ in Eq. (\ref{eq:walkOp}), which is proved in Appendix \ref{sec:proofWalkOp}.

\begin{theorem}
    Consider the Markov chain generated by Algorithm \ref{alg:MH} and denote by $\Delta$ its spectral gap.
    Define
    \begin{equation}
    \mathcal{A}:={\rm span}\{\ket{x}_{R_{\rm S}}\ket{0}_{R_{\rm M}}\ket{0}_{R_{\rm C}} | x\in\Omega\},\mathcal{B}:=V^\dagger B^\dagger SFBV\mathcal{A}. \label{eq:subspaceB}
    \end{equation}
    Then, on $\mathcal{A}+\mathcal{B}$, $\ket{P}$ is the unique eigenstate of $U$ with eigenvalue 1, and any other eigenvalue is written as $e^{i\theta}$ with $\theta\in\mathbb{R}$ such that $|\theta|\ge\arccos(1-\Delta)$.
    \label{th:phasegapOurs}
\end{theorem}

Let us consider how to implement the building-block operators in $U$.
$F$ and $S$ are an addition and a sign flip controlled by the qubit $R_{\rm C}$, respectively.
Various quantum circuits for arithmetic have been proposed so far (see \cite{MunosCoreas2022} as a review on circuits for four arithmetic operations and \cite{Bhaskar2016,haner2018} as studies on circuits for elementary functions), and making them controlled is straightforward.
$R$ is an operator that multiplies $-1$ to the state vector when all the qubits in $R_{\rm M}$ and $R_{\rm C}$ take $\ket{0}$ and thus implemented with a multi-controlled Pauli $Z$ gate.
$V$ is a circuit for loading a probability distribution into a quantum state, which have been also studied widely so far.
If $T$ can be calculated by some arithmetic, $V$ can be implemented by so-called Grover-Rudolph method \cite{grover2002creating}, using a logarithmic number of arithmetic circuits with respect to the number of grid points for discrete approximation.
Recently, some methods that avoid usage of arithmetic circuits have been proposed \cite{Sanders2019,wang2021fast,Wang_2022,rattew2022preparing,Bausch2022fastblackboxquantum}, including variational ones such as quantum generative adversarial network \cite{DallaireDemers2018,zoufal2019quantum,SITU2020193,Stein2021,Anand2021,assouel2022quantum,Agliardi2022,kasture2022protocols}.

Compared with these operators, $B$ can be costly in some situations.
Specifically, calculating the target distribution $P$, which is needed to evaluate the acceptance ratio, can be costly.
For example, parameter estimation in GW detection experiments, which has been mentioned in Introduction and will be explained in more detail in Sec. \ref{sec:GW}, $P$ is obtained via calculating the log-likelihood function.
It is determined by GW parameters and detector output data and evaluated as a sum of many terms that corresponds to contributions from various frequency modes of the data.
Naively calculating and summing up these terms leads to a large number of operations proportional to the number of terms.
More generally, a similar issue can arise in big-data analysis, specifically, when we estimate parameters of a statistical model based on a lot of independent sample data and the log-likelihood is a sum of contributions from them. 

\subsection{Approximate quantum walk operator via calculating the target distribution by quantum Monte Carlo integration \label{sec:walkOpQMCI}}

Then, we are motivated to develop some faster way to calculate of $P$ in the aforementioned situation.
We consider whether QMCI can be used to speedup summation of many terms in calculation of $P$.

We start from presenting the setup we consider.
We make the following assumption.

\begin{assumption}
For every $x\in\Omega$, $P$ is written as
\begin{equation}
    P(x)=P_0(x)e^{-L(x)}. \label{eq:PeL}
\end{equation}
Here, $P_0$ is a probability distribution on $\Omega$.
$L:\Omega\rightarrow\mathbb{R}_+$ is called the negative log-likelihood and written as 
\begin{equation}
    L(x)=L_{\rm sum}(x) + \ell_0(x)+C \label{eq:likeli}
\end{equation}
with $\ell_0:\Omega\rightarrow\mathbb{R}$, $C$ a constant independent of $x$, and
\begin{equation}
    L_{\rm sum}(x):=\frac{1}{M}\sum_{i=0}^{M-1} \ell(i,x), \label{eq:Lsum}
\end{equation}
where $M\in\mathbb{N}$ and $\ell:[M]_0\times\Omega\rightarrow\mathbb{R}$.
Besides, we are given the quantum circuits $O_{\ell}$, which acts on a 3-register system as
\begin{equation}
    O_{\ell}\ket{x}\ket{i}\ket{0}=\ket{x}\ket{i}\ket{\ell(i,x)}
\end{equation}
for any $i\in[M]_0$ and $x\in\Omega$.
Moreover, we are given $\sigma\in\mathbb{R}_+$ such that
\begin{equation}
    \frac{1}{M}\sum_{i=0}^{M-1}(\ell(i,x))^2 - \left(\frac{1}{M}\sum_{i=0}^{M-1}\ell(i,x)\right)^2 \le \sigma^2
\end{equation}
for any $x\in\Omega$.
\label{ass:P}
\end{assumption}

This assumption is threefold.
The first part, the form of $P$ is in line with the aforementioned situation, where the log-likelihood contains a sum of many terms.
The second one is availability of the quantum circuit $O_{\ell}$ to calculate the terms $\ell$, which is used in QMCI.
For large $M$, $O_{\ell}$ is the circuit queried most, and thus we hereafter focus on the number of queries to this as a metric of the complexity of our algorithm. 
The third one, the boundedness of the variance of $\ell$, is needed to bound the error in QMCI.

Note that the form of $L_{\rm sum}$ is in fact an average rather than a sum.
This is just for convenience in applying QMCI to computing it.
Also note that the order of $\sigma$ can depends on the term number $M$.
For example, if $L_{\rm sum}$ is a sum of contributions from $M$ independent samples, which applies to many cases in estimating parameters of statistical models, putting an overall factor $1/M$ and redefining $M\ell$ as $\ell$ leads to the form in Eq. (\ref{eq:Lsum}), but this makes the order of $\ell$ $O(M)$ if it is originally independent of $M$.

Hereafter, we denote by $\mathcal{C}_L$ the Markov chain generated by Algorithm \ref{alg:MH} with $P$ written as Eq. (\ref{eq:PeL}) with $L$.

We also assume the availability of the quantum circuits to generate the states that encode the proposal distribution $T$ and the prior distribution $P_0$ in amplitudes.

\begin{assumption}
    We are given quantum circuits $V$ that acts as Eq. (\ref{eq:V}).
    \label{ass:oraclesForT}
\end{assumption}

\begin{assumption}
    We are given quantum circuits $O_{P_0}$ that acts as Eq. (\ref{eq:OP0}).
    \label{ass:OP0}
\end{assumption}

Furthermore, we assume that we can use a quantum circuit to compute the acceptance ratio $A(x,y)$, given estimates $\hat{L}_x$ and $\hat{L}_y$ of $L_{\rm sum}(x)$ and $L_{\rm sum}(y)$.

\begin{assumption}
    We are given quantum circuits $O_{\rm AR}$ that acts as
    \begin{align}
    & O_{\rm AR}\ket{x}\ket{y}\ket{\hat{L}_x}\ket{\hat{L}_y}\ket{0}= \nonumber \\
    & \quad \ket{x}\ket{y}\ket{\hat{L}_x}\ket{\hat{L}_y}\Ket{\frac{P_0(y)T(y,x)\exp\left(-\left(\hat{L}_y+\ell_0(y)\right)\right)}{P_0(x)T(x,y)\exp\left(-\left(\hat{L}_x+\ell_0(x)\right)\right)}}
    \end{align}
    for any $x,y\in\Omega$ and $\hat{L}_x,\hat{L}_y\in\mathbb{R}$.
    \label{ass:OAR}
\end{assumption}

In many cases, the formulae for $P_0$ and $T$ are explicitly given with elementary functions and thus $O_{\rm AR}$ is implemented with arithmetic circuits.

Under these assumptions, Theorem \ref{th:QMCI} leads to the following lemma.

\begin{lemma}
Let $\Omega$ be a finite subset of $\mathbb{R}^d$ and $P$ be a distribution on it.
Under Assumptions \ref{ass:P} and \ref{ass:OAR}, for any $\delta,\epsilon\in(0,1)$, we have an access to a unitary operator $\tilde{\tilde{B}}_{\delta,\epsilon}$ on a system of three registers $R_{\rm S}, R_{\rm M}$ and $R_{\rm A}$ and a qubit $R_{\rm C}$ such that, for any $x\in\Omega$, $\Delta x\in\Delta\Omega_x$ and state $\ket{\phi}_{R_{\rm C}}$ on $R_{\rm C}$,
\begin{align}
    &\tilde{\tilde{B}}_{\delta,\epsilon}\ket{x}_{R_{\rm S}}\ket{\Delta x}_{R_{\rm M}}\ket{\phi}_{R_{\rm C}}\ket{0}_{R_{\rm A}}=  \nonumber \\
    & \quad \tilde{B}_{\epsilon}\left(\ket{x}_{R_{\rm S}}\ket{\Delta x}_{R_{\rm M}}\ket{\phi}_{R_{\rm C}}\right)\ket{0}_{R_{\rm A}} + \gamma_{x,\Delta x,\delta,\epsilon}\ket{\Psi}.
    \label{eq:Btilde}
\end{align}
Here,
\begin{align}
    & \tilde{B}_{\epsilon}\ket{x}_{R_{\rm S}}\ket{\Delta x}_{R_{\rm M}}\ket{\phi}_{R_{\rm C}}= \nonumber\\
    & \ \ \ket{x}_{R_{\rm S}}\ket{\Delta x}_{R_{\rm M}} \nonumber \\
    & \ \ \otimes
    \begin{pmatrix}
    \sqrt{1-\tilde{A}(x,x+\Delta x)} & -\sqrt{\tilde{A}(x,x+\Delta x)} \\
    \sqrt{\tilde{A}(x,x+\Delta x)} & \sqrt{1-\tilde{A}(x,x+\Delta x)}
    \end{pmatrix}
    \ket{\phi}_{R_{\rm C}},
\end{align}
where $\tilde{A}:\Omega\times\Omega\rightarrow\mathbb{R}$ is written as
\begin{equation}
    \tilde{A}(x,y)=\min\left\{1,\frac{P_0(y)e^{-\tilde{L}(y)}T(y,x)}{P_0(x)e^{-\tilde{L}(x)}T(x,y)}\right\} \label{eq:Atil}
\end{equation}
with $\tilde{L}:\Omega\rightarrow\mathbb{R}_+$ such that
\begin{equation}
    \max_{x\in\Omega} |\tilde{L}(x)-L(x)|\le\epsilon. \label{eq:delLeps}
\end{equation}
$\ket{\Psi}$ is some state on the entire system. $\gamma_{x,\Delta x,\delta,\epsilon}\in\mathbb{C}$ satisfies $|\gamma_{x,\Delta x,\delta,\epsilon}|\le \delta$.
$\tilde{\tilde{B}}_{\delta,\epsilon}$ makes queries to $O_{\ell}$, whose number is of order (\ref{eq:compQMCI}).
The qubit number in the entire system is of order (\ref{eq:qubitQMCI}).

\label{lem:approxB}
\end{lemma}

\begin{proof}

Because of Theorem \ref{th:QMCI}, given $O_{\ell}$, we have an access to an unitary operator $O_{L_{\rm sum},\epsilon,\frac{\delta^2}{16},\sigma}$ on the system of three registers that, for any $x\in\Omega$, acts as
\begin{equation}
    O_{L_{\rm sum},\epsilon,\frac{\delta^2}{16},\sigma}\ket{x}\ket{0}\ket{0}= \ket{x}\left(\ket{0}\ket{\tilde{L}_{\rm sum}(x)} + \gamma_{x,\frac{\delta^2}{16}}\ket{\psi_x}\right), \label{O_L}
\end{equation}
where $\gamma_{x,\frac{\delta^2}{16}}\in\mathbb{C}$ satisfies $\left|\gamma_{x,\frac{\delta^2}{16}}\right|^2\le\frac{\delta^2}{16}$, $\ket{\psi_x}$ is some state on the system of the second and third register, and $\tilde{L}_{\rm sum}:\Omega\rightarrow\mathbb{R}$ satisfies
\begin{equation}
    \max_{x\in\Omega} |\tilde{L}_{\rm sum}(x)-L_{\rm sum}(x)|\le\epsilon. \label{eq:delLsumeps}
\end{equation}
Equipped with this, we can construct the quantum circuit for the following operation on the system of $R_{\rm S}, R_{\rm M},R_{\rm C}$ and ancillary registers $R_{\rm A,1},...,R_{\rm A,6}$:
\begin{widetext}
\begin{align}
    & \ket{x}_{R_{\rm S}}\ket{\Delta x}_{R_{\rm M}}\ket{0}_{R_{\rm A,1}}\ket{0}_{R_{\rm A,2}}\ket{0}_{R_{\rm A,3}}\ket{0}_{R_{\rm A,4}}\ket{0}_{R_{\rm A,5}}\ket{0}_{R_{\rm A,6}}\ket{\phi}_{R_{\rm C}} \nonumber \\
     \rightarrow & \ket{x}_{R_{\rm S}}\ket{\Delta x}_{R_{\rm M}}\ket{x+\Delta x}_{R_{\rm A,1}}\ket{0}_{R_{\rm A,2}}\ket{0}_{R_{\rm A,3}}\ket{0}_{R_{\rm A,4}}\ket{0}_{R_{\rm A,5}}\ket{0}_{R_{\rm A,6}}\ket{\phi}_{R_{\rm C}}  \nonumber \\
     \rightarrow & \ket{x}_{R_{\rm S}}\ket{\Delta x}_{R_{\rm M}}\ket{x+\Delta x}_{R_{\rm A,1}} \left(\ket{0}_{R_{\rm A,2}}\ket{\tilde{L}_{\rm sum}(x)}_{R_{\rm A,3}} + \gamma_{x,\frac{\delta^2}{16}}\ket{\psi_x}_{R_{\rm A,2},R_{\rm A,3}}\right)\ket{0}_{R_{\rm A,4}}\ket{0}_{R_{\rm A,5}}\ket{0}_{R_{\rm A,6}}\ket{\phi}_{R_{\rm C}} \nonumber \\
    :=& \ket{x}_{R_{\rm S}}\ket{\Delta x}_{R_{\rm M}}\ket{x+\Delta x}_{R_{\rm A,1}}\ket{0}_{R_{\rm A,2}}\ket{\tilde{L}_{\rm sum}(x)}_{R_{\rm A,3}}\ket{0}_{R_{\rm A,4}}\ket{0}_{R_{\rm A,5}}\ket{0}_{R_{\rm A,6}}\ket{\phi}_{R_{\rm C}} + \gamma_{x,\frac{\delta^2}{16}}\ket{\Psi^{(1)}_{x,\Delta x}} \nonumber \\
     \rightarrow & \ket{x}_{R_{\rm S}}\ket{\Delta x}_{R_{\rm M}}\ket{x+\Delta x}_{R_{\rm A,1}} \ket{0}_{R_{\rm A,2}}\ket{\tilde{L}_{\rm sum}(x)}_{R_{\rm A,3}}\left(\ket{0}_{R_{\rm A,4}}\ket{\tilde{L}_{\rm sum}(x+\Delta x)}_{R_{\rm A,5}} + \gamma_{x+\Delta x,\frac{\delta^2}{16}}\ket{\psi_{x+\Delta x}}_{R_{\rm A,4},R_{\rm A,5}}\right)\ket{0}_{R_{\rm A,6}}\ket{\phi}_{R_{\rm C}} \nonumber \\
     & \quad + \gamma_{x,\frac{\delta^2}{16}}\ket{\Psi^{(2)}_{x,\Delta x}} \nonumber \\
    :=& \ket{x}_{R_{\rm S}}\ket{\Delta x}_{R_{\rm M}}\ket{x+\Delta x}_{R_{\rm A,1}}\ket{0}_{R_{\rm A,2}}\ket{\tilde{L}_{\rm sum}(x)}_{R_{\rm A,3}}\ket{0}_{R_{\rm A,4}}\ket{\tilde{L}_{\rm sum}(x+\Delta x)}_{R_{\rm A,5}}\ket{0}_{R_{\rm A,6}}\ket{\phi}_{R_{\rm C}} + \gamma^\prime_{x,\Delta x}\ket{\Psi^{(3)}_{x,\Delta x}} \nonumber \\
     \rightarrow & \ket{x}_{R_{\rm S}}\ket{\Delta x}_{R_{\rm M}}\ket{x+\Delta x}_{R_{\rm A,1}}\ket{0}_{R_{\rm A,2}}\ket{\tilde{L}_{\rm sum}(x)}_{R_{\rm A,3}}\ket{0}_{R_{\rm A,4}}\ket{\tilde{L}_{\rm sum}(x+\Delta x)}_{R_{\rm A,5}}\Ket{\tilde{A}(x,x+\Delta x)}_{R_{\rm A,6}}\ket{\phi}_{R_{\rm C}}
    + \gamma^\prime_{x,\Delta x}\ket{\Psi^{(4)}_{x,\Delta x}} \nonumber \\
    \rightarrow&\ket{x}_{R_{\rm S}}\ket{\Delta x}_{R_{\rm M}}\ket{x+\Delta x}_{R_{\rm A,1}}\ket{0}_{R_{\rm A,2}}\ket{\tilde{L}_{\rm sum}(x)}_{R_{\rm A,3}}\ket{0}_{R_{\rm A,4}}\ket{\tilde{L}_{\rm sum}(x+\Delta x)}_{R_{\rm A,5}}\ket{\tilde{A}(x,x+\Delta x)}_{R_{\rm A,6}} \nonumber \\
    &\quad\otimes
    \begin{pmatrix}
    \sqrt{1-\tilde{A}(x,x+\Delta x)} & -\sqrt{\tilde{A}(x,x+\Delta x)} \\
    \sqrt{\tilde{A}(x,x+\Delta x)} & \sqrt{1-\tilde{A}(x,x+\Delta x)}
    \end{pmatrix}
    \ket{\phi}_{R_{\rm C}} + \gamma^\prime_{x,\Delta x}\ket{\Psi^{(5)}_{x,\Delta x}}\nonumber \\
    \rightarrow & \ket{x}_{R_{\rm S}}\ket{\Delta x}_{R_{\rm M}}\ket{x+\Delta x}_{R_{\rm A,1}}\ket{0}_{R_{\rm A,2}}\ket{\tilde{L}_{\rm sum}(x)}_{R_{\rm A,3}}\ket{0}_{R_{\rm A,4}}\ket{\tilde{L}_{\rm sum}(x+\Delta x)}_{R_{\rm A,5}}\ket{0}_{R_{\rm A,6}} \nonumber \\
    & \quad\otimes
    \begin{pmatrix}
    \sqrt{1-\tilde{A}(x,x+\Delta x)} & -\sqrt{\tilde{A}(x,x+\Delta x)} \\
    \sqrt{\tilde{A}(x,x+\Delta x)} & \sqrt{1-\tilde{A}(x,x+\Delta x)}
    \end{pmatrix}
    \ket{\phi}_{R_{\rm C}} + \gamma^\prime_{x,\Delta x}\ket{\Psi^{(6)}_{x,\Delta x}} \nonumber \\
    \rightarrow & \ket{x}_{R_{\rm S}}\ket{\Delta x}_{R_{\rm M}}\ket{x+\Delta x}_{R_{\rm A,1}}\ket{0}_{R_{\rm A,2}}\ket{0}_{R_{\rm A,3}}\ket{0}_{R_{\rm A,4}}\ket{0}_{R_{\rm A,5}}\ket{0}_{R_{\rm A,6}} \nonumber \\
    & \quad \otimes
    \begin{pmatrix}
    \sqrt{1-\tilde{A}(x,x+\Delta x)} & -\sqrt{\tilde{A}(x,x+\Delta x)} \\
    \sqrt{\tilde{A}(x,x+\Delta x)} & \sqrt{1-\tilde{A}(x,x+\Delta x)}
    \end{pmatrix}
    \ket{\phi}_{R_{\rm C}} + \gamma^{\prime\prime}_{x,\Delta x}\ket{\Psi^{(7)}_{x,\Delta x}}  \nonumber \\
    \rightarrow & \ket{x}_{R_{\rm S}}\ket{\Delta x}_{R_{\rm M}}\ket{0}_{R_{\rm A,1}}\ket{0}_{R_{\rm A,2}}\ket{0}_{R_{\rm A,3}}\ket{0}_{R_{\rm A,4}}\ket{0}_{R_{\rm A,5}}\ket{0}_{R_{\rm A,6}} \nonumber \\
    & \quad \otimes
    \begin{pmatrix}
    \sqrt{1-\tilde{A}(x,x+\Delta x)} & -\sqrt{\tilde{A}(x,x+\Delta x)} \\
    \sqrt{\tilde{A}(x,x+\Delta x)} & \sqrt{1-\tilde{A}(x,x+\Delta x)}
    \end{pmatrix}
    \ket{\phi}_{R_{\rm C}} + \gamma^{\prime\prime}_{x,\Delta x}\ket{\Psi^{(8)}_{x,\Delta x}}  \nonumber \\
    =:&\ket{\tilde{\Phi}_{x,\Delta x}}
    , \label{eq:stFlow}
\end{align}
\end{widetext}
where $\ket{\Psi^{(1)}_{x,\Delta x}}$, ..., $\ket{\Psi^{(8)}_{x,\Delta x}}$ are some states on the entire system and $\gamma^{\prime}_{x,\Delta x},\gamma^{\prime\prime}_{x,\Delta x}\in\mathbb{C}$.
In Eq. (\ref{eq:stFlow}), we use an adder circuit at the first arrow.
At the second and third arrows, we use $O_{L_{\rm sum},\epsilon,\frac{\delta^2}{16},\sigma}$ on the system of $R_{\rm S}$, $R_{\rm A,2}$ and $R_{\rm A,3}$ and that of $R_{\rm A_1}$, $R_{\rm A,4}$ and $R_{\rm A,5}$, respectively.
At the fourth arrow, we use $O_{\rm AR}$ to compute $\tilde{A}(x,x+\Delta x)$ as Eq. (\ref{eq:Atil}) with $\tilde{L}=\tilde{L}_{\rm sum}+\ell_0+C$, which satisfies Eq. (\ref{eq:delLeps}) because of Eq. (\ref{eq:delLsumeps}).
The fifth arrow is by the Y-rotation $\begin{pmatrix} \cos\frac{\varphi}{2} & -\sin\frac{\varphi}{2} \\ \sin\frac{\varphi}{2} & \cos\frac{\varphi}{2} \end{pmatrix}$ on $R_{\rm C}$ with the rotation angle specified by $R_{\rm A,6}$, which is implemented as follows \cite{Egger2021}: we compute $\varphi=2\arcsin \left(\sqrt{\tilde{A}(x,x+\Delta x)}\right)$ onto another ancillary register using arithmetic circuits \cite{MunosCoreas2022,haner2018,Bhaskar2016} and apply fixed-angle Y-rotation gates controlled by qubits in that ancillary register to $R_{\rm C}$.
At the sixth arrow, we perform the inverse of the operation at the fourth arrow.
The seventh arrow is by the inverses of the operations at the second and third arrows, which act as
\begin{align}
    & (O_{L_{\rm sum},\epsilon,\frac{\delta^2}{16},\sigma})^\dagger\ket{x}_{R_{\rm S}}\ket{0}_{R_{\rm A,2}}\ket{\tilde{L}_{\rm sum}(x)}_{R_{\rm A,3}} \nonumber \\
    =& (O_{L_{\rm sum},\epsilon,\frac{\delta^2}{16},\sigma})^\dagger\ket{x}_{R_{\rm S}} \nonumber \\
    & \quad \otimes\left(\ket{0}_{R_{\rm A,2}}\ket{\tilde{L}_{\rm sum}(x)}_{R_{\rm A,3}}+ \gamma_{x,\frac{\delta^2}{16}}\ket{\psi_x}_{R_{\rm A,2},R_{\rm A,3}}\right) \nonumber \\
    & \quad -\gamma_{x,\frac{\delta^2}{16}}(O_{L_{\rm sum},\epsilon,\frac{\delta^2}{16},\sigma})^\dagger\ket{\psi_x}_{R_{\rm A,2},R_{\rm A,3}}\ \nonumber \\
    =& \ket{x}_{R_{\rm S}}\ket{0}_{R_{\rm A,2}}\ket{0}_{R_{\rm A,3}} -\gamma_{x,\frac{\delta^2}{16}}(O_{L_{\rm sum},\epsilon,\frac{\delta^2}{16},\sigma})^\dagger\ket{\psi_x}_{R_{\rm A,2},R_{\rm A,3}}
\end{align}
and, similarly,
\begin{align}
    & (O_{L_{\rm sum},\epsilon,\frac{\delta^2}{16},\sigma})^\dagger\ket{x+\Delta x}_{R_{\rm A,1}}\ket{0}_{R_{\rm A,4}}\ket{\tilde{L}_{\rm sum}(x+\Delta x)}_{R_{\rm A,5}} \nonumber \\
    =& \ket{x+\Delta x}_{R_{\rm A,1}}\ket{0}_{R_{\rm A,4}}\ket{0}_{R_{\rm A,5}} \nonumber \\
    & \ -\gamma_{x+\Delta x,\frac{\delta^2}{16}}(O_{L_{\rm sum},\epsilon,\frac{\delta^2}{16},\sigma})^\dagger\ket{\psi_{x+\Delta x}}_{R_{\rm A,4},R_{\rm A,5}}.
\end{align}
At the last arrow, we perform the inverse of the operation at the first arrow.

Then, let us show that $\ket{\tilde{\Phi}_{x,\Delta x}}$ is in the form of Eq. (\ref{eq:Btilde}), with $R_{\rm A,1},...,R_{\rm A,6}$ collectively seen as $R_{\rm A}$.
Since we have seen that Eq. (\ref{eq:delLeps}) holds, it is sufficient to check $|\gamma^{\prime\prime}_{x,\Delta x}|\le\delta$.
This is done as follows.
We see that $\gamma^\prime_{x,\Delta x}$, which is introduced by two $O_{L_{\rm sum},\epsilon,\frac{\delta^2}{16},\sigma}$'s in the second and third arrows  in Eq. (\ref{eq:stFlow}), is bounded as
\begin{equation}
    \left|\gamma^\prime_{x,\Delta x}\right| \le \left|\gamma_{x,\frac{\delta^2}{16}}\right| + \left|\gamma_{x+\Delta x,\frac{\delta^2}{16}}\right|\le \sqrt{\frac{\delta^2}{16}} + \sqrt{\frac{\delta^2}{16}} = \frac{\delta}{2}.
\end{equation}
Similarly, applying $(O_{L_{\rm sum},\epsilon,\frac{\delta^2}{16},\sigma})^\dagger$ twice at the eighth arrow in Eq. (\ref{eq:stFlow}) increases this by at most $\frac{\delta}{2}$:
\begin{equation}
|\gamma^{\prime\prime}_{x,\Delta x}|
\le |\gamma^\prime_{x,\Delta x}|+\frac{\delta}{2}  \le \delta. \label{eq:gamma2pr}
\end{equation}
Thus, we have $|\gamma^{\prime\prime}_{x,\Delta x}|\le\delta$.

The statements on the number of queries to $O_\ell$ and the qubit number immediately follow from Theorem \ref{th:QMCI}, which gives the bounds on the query number and qubit number in $O_{L_{\rm sum},\epsilon,\frac{\delta^2}{16},\sigma}$ as Eqs. (\ref{eq:compQMCI}) and (\ref{eq:qubitQMCI}).

\end{proof}

We now define the approximate quantum walk operator.
\begin{equation}
    \tilde{\tilde{U}}_{\delta,\epsilon}:=\tilde{R}\tilde{V}^\dagger \tilde{\tilde{B}}_{\frac{\delta}{2},\epsilon}^\dagger \tilde{S}\tilde{F}\tilde{\tilde{B}}_{\frac{\delta}{2},\epsilon}\tilde{V} \label{eq:walkOpApp}
\end{equation}
on the system of $R_{\rm S}$, $R_{\rm M}$, $R_{\rm A}$ and $R_{\rm C}$, with $\tilde{R}$, $\tilde{V}$, $\tilde{S}$ and $\tilde{F}$ defined as $R\otimes I_{R_{\rm A}}$ and so on.
We then have the following lemma immediately.

\begin{lemma}
Let $\Omega$ be a finite subset of $\mathbb{R}^d$ and $P$ be a distribution on it.
Under Assumption \ref{ass:P}, \ref{ass:oraclesForT} and \ref{ass:OAR}, for any 
$\delta,\epsilon\in(0,1)$, we have access to a unitary operator $\tilde{\tilde{U}}_{\delta,\epsilon}$ on a system of three registers $R_{\rm S}, R_{\rm M}$ and $R_{\rm A}$ and a qubit $R_{\rm C}$, which is given as Eq. (\ref{eq:walkOpApp}), and, for any $x\in\Omega$, $\Delta x\in\Delta\Omega_x$ and state $\ket{\phi}_{R_{\rm C}}$ on $R_{\rm C}$, acts as
\begin{align}
    &  \tilde{\tilde{U}}_{\delta,\epsilon}\ket{x}_{R_{\rm S}}\ket{\Delta x}_{R_{\rm M}}\ket{\phi}_{R_{\rm C}}\ket{0}_{R_{\rm A}} \nonumber\\
    &  \ = \left(\tilde{U}_{\epsilon}\ket{x}_{R_{\rm S}}\ket{\Delta x}_{R_{\rm M}}\ket{\phi}_{R_{\rm C}}\right)\ket{0}_{R_{\rm A}}+ \tilde{\gamma}_{x,\Delta x,\phi,\delta,\epsilon}\ket{\tilde{\Psi}},
    \label{eq:Utilde}
\end{align}
where $\ket{\tilde{\Psi}}$ is some state on the entire system, $\tilde{\gamma}_{x,\Delta x,\phi,\delta,\epsilon}\in\mathbb{C}$ satisfies $|\tilde{\gamma}_{x,\Delta x,\phi,\delta,\epsilon}|\le \delta$ and $\tilde{U}_{\epsilon}:=RV^\dagger \tilde{B}_{\epsilon}^\dagger SF\tilde{B}_{\epsilon}V$.
$\tilde{\tilde{U}}_{\delta,\epsilon}$ makes a number of order (\ref{eq:compQMCI}) of calls to $O_{\ell}$ and uses a number of order (\ref{eq:qubitQMCI}) of qubits.

\label{lem:approxQWalk}
\end{lemma}

Note that $\tilde{\tilde{U}}_{\delta,\epsilon}$ has two types of differences from the exact quantum walk operator $U$.
First, $\tilde{\tilde{U}}_{\delta,\epsilon}$ does not exactly act as a quantum walk operator because it generates the residual term $\tilde{\gamma}_{x,\Delta x,\phi,\delta,\epsilon}\ket{\tilde{\Psi}}$ in Eq. (\ref{eq:Utilde}).
Second, even if there were no residual term, $\tilde{\tilde{U}}_{\delta,\epsilon}$ would not be the quantum walk operator for the Markov chain $\mathcal{C}_L$ we consider but that for another one $\mathcal{C}_{\tilde{L}}$ because of the error in the approximation $\tilde{L}$ of the exact negative log-likelihood $L$.
This difference makes the stationary distribution differ from the target distribution $P$.
Nevertheless, we can use $\tilde{\tilde{U}}_{\delta,\epsilon}$, controlling these differences by taking sufficiently small $\delta$ and $\epsilon$.   

\subsection{Quantum simulated annealing with the approximate quantum walk operator \label{sec:appQSA}}

Now, we can construct an approximation of the phase gate $R^\omega_{\ket{P}}$ using this $\tilde{\tilde{U}}_{\delta,\epsilon}$ instead of the exact quantum walk operator $U$.

\begin{lemma}

Let $\delta,\epsilon\in(0,1)$.
Under Assupmtions \ref{ass:P}, \ref{ass:oraclesForT} and \ref{ass:OAR}, consider a Markov chain $\mathcal{C}_L$.
Denote its transition matrix by $W$ and its spectral gap by $\Delta$.
Denote by $\kappa$ the condition number of the matrix $Q$ such that $Q^{-1}WQ$ is diagonal.
Let $\omega$ be a complex number with unit modulus.
Then, we have an access to a unitary operator $\tilde{\tilde{R}}^\omega_{L,\delta,\epsilon}$ that has the following properties:
\begin{itemize}
    
    \item $\tilde{\tilde{R}}^\omega_{L,\delta,\epsilon}$ acts on a system of $R_{\rm S}$ and $n_{\rm anc}$ ancillary qubits.
    Here,
    \begin{align}
    &n_{\rm anc}=O\left(\log\left(\frac{1}{\Delta}\right)\log\left(\frac{1}{\delta}\right)\right.+ \nonumber \\
    &\quad \left.\left(\log M+\log\left(\frac{\sigma}{\epsilon^\prime}\right)\right)\log\left(\frac{\sigma}{\epsilon^\prime}\right)\log\log\left(\frac{\sigma}{\epsilon^\prime}\right)\log\left(\frac{1}{\delta\sqrt{\Delta}}\right)\right),
    \label{eq:qubitRtiltil}
    \end{align}
    where
    \begin{equation}
        \epsilon^\prime:=\min\left\{\epsilon,\frac{\Delta}{16\sqrt{\max_{y\in\Omega} \sum_{x\in\Omega\setminus\{y\}} T_{xy}}\kappa}\right\}.
    \end{equation}
    
    \item $\tilde{\tilde{R}}^\omega_{L,\delta,\epsilon}$ uses $O_{\ell}$
    \begin{equation}
        O\left(\frac{\sigma}{\epsilon^\prime\sqrt{\Delta}}\log^{3/2}\left(\frac{\sigma}{\epsilon^\prime}\right)\log\log\left(\frac{\sigma}{\epsilon^\prime}\right)\log\left(\frac{1}{\delta\sqrt{\Delta}}\right)\log\left(\frac{1}{\delta}\right)\right) \label{eq:compUtil}
    \end{equation} times.
    
    \item For any state $\ket{\Xi}$ on $R_{\rm S}$,
    \begin{equation}
    \tilde{\tilde{R}}^{\omega}_{L,\delta,\epsilon}\ket{\Xi}\ket{0}^{\otimes n_{\rm anc}}=(R^{\omega}_{\ket{\tilde{P}}}\ket{\Xi})\ket{0}^{\otimes n_{\rm anc}}+\ket{\xi}. \label{eq:Rtiltil}
    \end{equation}
    Here, $R^{\omega}_{\ket{\tilde{P}}}$ is a unitary defined as Eq. (\ref{eq:RomegaP}), where $\tilde{P}$ is a distribution on $\Omega$ in the form of
    \begin{equation}
        \tilde{P}(x) \propto P_0(x)e^{-\tilde{L}(x)}
    \end{equation}
    with some function $\tilde{L}:\Omega\rightarrow\mathbb{R}_+$ satisfying
    \begin{equation}
    \max_{x\in\Omega} |\tilde{L}(x)-L(x)|\le\epsilon^\prime, \label{eq:delLEpsPr}
    \end{equation}
    and $\ket{\xi}$ is an unnormalized state vector with $\|\ket{\xi}\|\le \delta$.

\end{itemize}

\label{lem:Rtiltil}
\end{lemma}

To prove this, we use the following lemma on the spectral gap of the Markov chain with replacement of $L$ with $\tilde{L}$.
The proof of this is presented in Appendix. \ref{sec:proofLemSpGapPert}.

\begin{lemma}
    Let $\tilde{L}:\Omega\rightarrow\mathbb{R}_+$ be a function on $\Omega$ and denote by $\tilde{\Delta}$ the spectral gap of $\mathcal{C}_{\tilde{L}}$.
    Let $\Delta$ and $\kappa$ be the same as Lemma \ref{lem:Rtiltil}.
    Then, if $\epsilon:=\max_{x\in\Omega} |\tilde{L}(x)-L(x)|\le \frac{1}{4}$,
    \begin{equation}
        \tilde{\Delta} \ge \Delta - 16\left(\max_{y\in\Omega} \sum_{x\in\Omega\setminus\{y\}} T_{xy}\right)^{1/2}\kappa\epsilon
        \label{eq:SpGapPert}
    \end{equation}
    holds.

    \label{lem:SpGapPert}
\end{lemma}

This lemma means that, although the MH Markov chain with $\tilde{L}$ instead of $L$ is different from the original one, the change of the spectral gap is small if $\tilde{L}$ is close to $L$.

Then, the proof of Lemma \ref{lem:Rtiltil} is as follows.

\begin{proof}[Proof of Lemma \ref{lem:Rtiltil}]

   Note that, because of Lemma \ref{lem:SpGapPert} and Eq. (\ref{eq:delLEpsPr}), the spectral gap $\tilde{\Delta}$ of $\mathcal{C}_{\tilde{L}}$ satisfies $\tilde{\Delta} \ge \frac{\Delta}{2}$.
    Then, because of Theorem \ref{th:Romega}, if we had an access to $\tilde{U}_{\epsilon^\prime}$, we could construct $\tilde{R}^\omega_{\ket{\tilde{P}},\frac{\delta}{2}}$ making $O\left(\frac{\log\left(1/\delta\right)}{\sqrt{\Delta}}\right)$ uses of $\tilde{U}_{\epsilon^\prime}$, with $\tilde{P}$ having the stated property.
    In reality, we can use $\tilde{\tilde{U}}_{\delta^\prime,\epsilon^\prime}$, an approximation of $\tilde{U}_{\epsilon^\prime}$ with some $\delta^\prime\in(0,1)$.
    Recalling Lemma \ref{lem:approxQWalk}, we see that the unitary $\mathcal{R}_{\delta^\prime}$ we obtain by using $\tilde{\tilde{U}}_{\delta^\prime,\epsilon^\prime}$ instead of $\tilde{U}_{\epsilon^\prime}$ in construction of $\tilde{R}^\omega_{\ket{\tilde{P}},\frac{\delta}{2}}$ acts as $\mathcal{R}_{\delta^\prime}\ket{\Xi}\ket{0}^{\otimes n_{\rm anc}}=\tilde{R}^\omega_{\ket{\tilde{P}},\frac{\delta}{2}}\ket{\Xi}\ket{0}^{\otimes n_{\rm anc}}+\ket{\xi^\prime}$, where $\|\ket{\xi^\prime}\|=O\left(\delta^\prime \frac{\log\left(1/\delta\right)}{\sqrt{\Delta}}\right)$.
    Thus, there exists $\delta^\prime=\Theta\left(\frac{\delta\sqrt{\Delta}}{\log\left(1/\delta\right)}\right)$ that makes $\|\ket{\xi^\prime}\|\le\frac{\delta}{2}$.
    Since $\left\|\tilde{R}^{\omega}_{\ket{\tilde{P}},\frac{\delta}{2}}\ket{\Xi}\ket{0}^{\otimes n_{\rm anc}}-(R^{\omega}_{\ket{\tilde{P}}}\ket{\Xi})\ket{0}^{\otimes n_{\rm anc}}\right\|\le\frac{\delta}{2}$, $\mathcal{R}_{\delta^\prime}$ with this $\delta^\prime$ is in fact $\tilde{\tilde{R}}^{\omega}_{L,\delta,\epsilon}$ that satisfies Eq. (\ref{eq:Rtiltil}).

    The statement on the qubit number follows since constructing $\tilde{R}^\omega_{\ket{\tilde{P}},\frac{\delta}{2}}$ with $\tilde{U}_{\epsilon^\prime}$ uses $O\left(\log\left(\frac{1}{\Delta}\right)\log\left(\frac{1}{\delta}\right)\right)$ qubits and using $\tilde{\tilde{U}}_{\delta^\prime,\epsilon^\prime}$ instead of $\tilde{U}_{\epsilon^\prime}$ adds $O\left(\left(\log M+\log\left(\frac{\sigma}{\epsilon^\prime}\right)\right)\log\left(\frac{\sigma}{\epsilon^\prime}\right)\log\log\left(\frac{\sigma}{\epsilon^\prime}\right)\log(\delta^\prime)^{-1}\right)$ qubits, whose sum is of order (\ref{eq:qubitRtiltil}).

    The upper bound (\ref{eq:compUtil}) on the number of queries to $O_{\ell}$ is obtained by substituting  $\delta^\prime=\Theta\left(\frac{\delta\sqrt{\Delta}}{\log\left(1/\delta\right)}\right)$ for $\delta$ and $\epsilon^\prime$ for $\epsilon$ in Eq. (\ref{eq:compQMCI}), which yields the query number in one $\tilde{\tilde{U}}_{\delta^\prime,\epsilon^\prime}$, and multiplying $O\left(\frac{\log\left(1/\delta\right)}{\sqrt{\Delta}}\right)$.
\end{proof}

We can use this approximate phase gate instead of the exact one in QSA.
Before we make a statement on this approximate QSA, let us make some preparation.
First, we make the following assumptions.
\begin{assumption}
    There exists $\Delta_{\rm min}\in(0,1)$ such that, for any $\beta\in(0,1]$, the spectral gap of the Markov chain $\mathcal{C}_{\beta L}$ is equal to or larger than $\Delta_{\rm min}$. 
    \label{ass:SpGap}
\end{assumption}
\begin{assumption}
    There exists $\kappa_{\rm min}\in\mathbb{R}_+$ such that, for any $\beta\in(0,1]$, the condition number of the matrix $Q_\beta$ that diagonalizes the transition matrix $W_\beta$ for the Markov chain $\mathcal{C}_{\beta L}$, which means $Q^{-1}_\beta W_\beta Q_\beta$ is diagonal, is equal to or smaller than $\kappa_{\rm min}$. 
    \label{ass:CondNum}
\end{assumption}

Besides, we also present the following lemma, whose proof is presented in Appendix. \ref{sec:proofLemTVvsLerr}.

\begin{lemma}
    Consider the Markov chains $\mathcal{C}_L$ and $\mathcal{C}_{\tilde{L}}$ with $L:\Omega\rightarrow\mathbb{R}_+$ and $\tilde{L}:\Omega\rightarrow\mathbb{R}_+$.
    Then, for their stationary distributions $P\propto P_0 e^{-L}$ and $\tilde{P}\propto P_0 e^{-\tilde{L}}$,
    \begin{equation}
        \|\tilde{P}-P\|_{\rm TV} \le 8\left(\left\lceil\frac{\log(2\sqrt{P_{\rm min}})}{\log(1-\Delta)}\right\rceil+\frac{1}{\Delta}\right)\times \max_{x\in\Omega} |\tilde{L}(x)-L(x)| \label{eq:diffPtilP}
    \end{equation}
    holds, where $\Delta$ is the spectral gap of $\mathcal{C}_L$ and $P_{\rm min}:=\min_{x\in\Omega}P(x)$.
    \label{lem:TVvsLerr}
\end{lemma}

Then, we have the following theorem.

\begin{theorem}
    Suppose that Assumptions \ref{ass:P} to \ref{ass:CondNum} are satisfied.
    Then, for any $\epsilon,\delta\in(0,1)$, there exists an algorithm that makes
    \begin{align}
        & O\left(\frac{\sigma l_{\rm max}}{\epsilon^{\prime\prime}\sqrt{\Delta_{\rm min}}}\log^{3/2}\left(\frac{\sigma}{\epsilon^{\prime\prime}}\right)\log\log\left(\frac{\sigma}{\epsilon^{\prime\prime}}\right)\log\left(\frac{1}{\sqrt{\Delta_{\rm min}}}\right)\right. \nonumber \\
        &\qquad \left. \times \left(\log l_{\rm max}+\log L_{\rm max}\log\left(\frac{l_{\rm max}L_{\rm max}}{\delta}\right)\right) \right) \label{eq:compGetBetaApp}
    \end{align}
    queries to $O_{\ell}$, where
    \begin{align}
        &\epsilon^{\prime\prime}:=\min\left\{\frac{\Delta_{\rm min}\epsilon}{8\left(\Delta_{\rm min}\left\lceil\frac{\log(2\sqrt{P_{\rm min}})}{\log(1-\Delta_{\rm min})}\right\rceil+1\right)},\right. \nonumber \\
        &\qquad\qquad\quad\left.\frac{\Delta_{\rm min}}{16\sqrt{\max_{y\in\Omega} \sum_{x\in\Omega\setminus\{y\}} T_{xy}}\kappa_{\rm min}}, 
        \frac{\bar{L}}{2}\right\},
    \end{align}
    and, with probability at least $1-\delta$, outputs a sequence $\tilde{\beta}_0=0<\tilde{\beta}_1<\cdots<\tilde{\beta}_{l-1}<\tilde{\beta}_l=1$ with following properties:
    \begin{itemize}
        \item $l\le l_{\rm max}$

        \item Given this sequence, we have an unitary operator that generates a state $\widetilde{\ket{\tilde{P}}}$ $\epsilon$-close to $\ket{\tilde{P}}\ket{0}^{\otimes n_{\rm anc}}$, where $\tilde{P}$ is a probability distribution on $\Omega$ such that $\|\tilde{P}-P\|_{\rm TV}\le\epsilon$ and
        \begin{align}
        &n_{\rm anc}= \nonumber \\
        &\quad O\left(\log\left(\frac{1}{\Delta_{\rm min}}\right)\log\left(\frac{l_{\rm max}}{\epsilon^{\prime\prime}}\right)\right.+ \nonumber \\
        &\quad \qquad \left(\log M+\log\left(\frac{\sigma }{\epsilon^{\prime\prime}}\right)\right)\log\left(\frac{\sigma}{\epsilon^{\prime\prime}}\right) \log\log\left(\frac{\sigma}{\epsilon^{\prime\prime}}\right) \nonumber \\
        &\quad \qquad \left.\times\log\left(\frac{l_{\rm max}}{\epsilon^{\prime\prime}\sqrt{\Delta_{\rm min}}}\right)\right).  \label{eq:qubitQSAApp}
        \end{align}
        In that operator, $O_{\ell}$ is called
    \begin{align}
    & O\left(\frac{\sigma l_{\rm max}}{\epsilon^{\prime\prime} \sqrt{\Delta_{\rm min}}}\log^{3/2}\left(\frac{\sigma}{\epsilon^{\prime\prime}}\right)\log\log\left(\frac{\sigma}{\epsilon^{\prime\prime}}\right)\log\left(\frac{l_{\rm max}}{\epsilon \sqrt{\Delta_{\rm min}}}\right)\right. \nonumber \\
    &\left.\qquad\times\log^2\left(\frac{l_{\rm max}}{\epsilon}\right) \right) \label{eq:compQSAApp}
    \end{align}
    times.
        
    \end{itemize}

    \label{th:QSAbyQMCI}
\end{theorem}

For $\epsilon$ such that
\begin{align}
    \epsilon&\le\min\left\{\frac{1}{2\sqrt{\max_{y\in\Omega} \sum_{x\in\Omega\setminus\{y\}} T_{xy}}\kappa_{\rm min}},\frac{4\bar{L}}{\Delta_{\rm min}}\right\} \nonumber \\
    &\qquad \times\left(\Delta_{\rm min}\left\lceil\frac{\log(2\sqrt{P_{\rm min}})}{\log(1-\Delta_{\rm min})}\right\rceil+1\right) \label{eq:smallEpsCond},
\end{align}
Eq. (\ref{eq:compQSAApp}) becomes
\begin{equation}
    \tilde{O}\left(\frac{\sigma \bar{L}^{1/2}}{\epsilon \Delta_{\rm min}^{3/2}}\right).
    \label{eq:compQSAAppWoLog}
\end{equation}

\begin{proof}[Proof of Theorem \ref{th:QSAbyQMCI}]

First, note that, for any $\beta\in(0,1]$,  $\tilde{\tilde{R}}^{\omega}_{\beta L,\delta^\prime,\epsilon^{\prime\prime}}$ is equal to $\tilde{R}^{\omega}_{\ket{\tilde{P}_{\beta}},\delta^\prime}$, where $\delta^\prime\in(0,1)$ and $\tilde{P}_{\beta}$ is a distribution on $\Omega$ in the form of $\tilde{P}_{\beta}\propto P_0e^{-\beta\tilde{L}}$ with some function $\tilde{L}:\Omega\rightarrow \mathbb{R}_+$ satisfying
\begin{equation}
    \max_{x\in\Omega} |\tilde{L}(x)-L(x)| \le \epsilon^{\prime\prime}. \label{eq:tilLepsprpr}
\end{equation}
On the other hand, according to \cite{Harrow2020}, given $\tilde{R}^{\omega_{\pi/3}}_{\ket{\tilde{P}_{\beta}},\delta^\prime}$ and $\tilde{R}^{-1}_{\ket{\tilde{P}_{\beta}},\delta^\prime}$ with some $\delta^\prime=\Theta(1)$ for any $\beta\in(0,1]$, Algorithm \ref{alg:QSA} outputs the sequence $\tilde{\beta}_0=0<\tilde{\beta}_1<\cdots<\tilde{\beta}_{l-1}<\tilde{\beta}_l=1$ such that $l\le \tilde{l}_{\rm max}$ and $|\braket{\tilde{P}_{\tilde{\beta}_i} | \tilde{P}_{\tilde{\beta}_{i+1}}}|^2\ge \frac{9}{10e^2}$ with probability at least $1-\delta$, making $O\left(\tilde{l}_{\rm max}\log\tilde{l}_{\rm max}\right)$ uses of operators in $\left\{\tilde{R}^{\omega_{\pi/3}}_{\ket{\tilde{P}_{\beta}},\delta^\prime} \ \middle| \ \beta\in(0,1]\right\}$ and $O\left(\tilde{l}_{\rm max}\log \tilde{L}_{\rm max}\log \left(\frac{\tilde{l}_{\rm max}\tilde{L}_{\rm max}}{\delta}\right)\right)$ uses of operators in $\left\{\tilde{R}^{\omega_{-1}}_{\ket{\tilde{P}_{\beta}},\delta^\prime} \ \middle| \ \beta\in(0,1]\right\}$.
Here, $\tilde{l}_{\rm max}:=\mathbb{E}_{P_0}[\tilde{L}(x)]\log\left(\mathbb{E}_{P_0}[\tilde{L}(x)]\right)$ and $\tilde{L}_{\rm max}:=\max_{x\in\Omega} \tilde{L}(x)$, which are of order $O(l_{\rm max})$ and $O(L_{\rm max})$, respectively, because of Eq. (\ref{eq:tilLepsprpr}).
Since we can construct $\tilde{\tilde{R}}^{\omega}_{\beta L,\delta^\prime,\epsilon^{\prime\prime}}$ similarly to $\tilde{\tilde{R}}^{\omega}_{L,\delta^\prime,\epsilon^{\prime\prime}}$ using $O_\ell$, we can obtain the above $\left\{\tilde{\beta}_i\right\}$ by Algorithm \ref{alg:QSA}.
Because of Lemma \ref{lem:Rtiltil}, the number of queries to $O_\ell$ in $\tilde{\tilde{R}}^{\omega}_{\beta L,\delta^\prime,\epsilon^{\prime\prime}}$ is 
\begin{equation}
    O\left(\frac{\sigma}{\epsilon^{\prime\prime}\sqrt{\tilde{\Delta}_{\rm min}}}\log^{3/2}\left(\frac{\sigma}{\epsilon^{\prime\prime}}\right)\log\log\left(\frac{\sigma}{\epsilon^{\prime\prime}}\right)\log\left(\frac{1}{\sqrt{\tilde{\Delta}_{\rm min}}}\right)\right).
\end{equation}
Here, $\tilde{\Delta}_{\rm min}$ is a lower bound of the spectral gaps of $\{\mathcal{C}_{\beta\tilde{L}} \ | \ \beta\in(0,1] \}$, which satisfies $\tilde{\Delta}_{\rm min}\ge\frac{\Delta_{\rm min}}{2}$ because of Lemma \ref{lem:SpGapPert} and Eq. (\ref{eq:tilLepsprpr}).
Combining the above discussions, we see that in finding $\tilde{\beta}_1,...,\tilde{\beta}_{l-1}$ the total number of calls to $O_{\ell}$ is of order (\ref{eq:compGetBetaApp}).

After that, as shown in \cite{Harrow2020}, a series of A$\frac{\pi}{3}$AA generates $\widetilde{\ket{\tilde{P}}}$ $\epsilon$-close to $\ket{\tilde{P}}$, where $\tilde{P}:=\tilde{P}_1\propto P_0 e^{-\tilde{L}}$.
According to \cite{Harrow2020}, in this procedure, we makes $O\left(l \log\left(\frac{l}{\epsilon}\right)\right)$ uses of $\tilde{\tilde{R}}^{\omega}_{\tilde{\beta}_i L,\delta^{\prime\prime},\epsilon^{\prime\prime}}$ with some $\delta^{\prime\prime}=\Theta\left(\frac{\epsilon}{l \log\left(l/\epsilon\right)}\right)$.
The number of calls to $O_\ell$ in $\tilde{\tilde{R}}^{\omega}_{\tilde{\beta}_i L,\delta^{\prime\prime},\epsilon^{\prime\prime}}$ is
\begin{equation}
    O\left(\frac{\sigma}{\epsilon^{\prime\prime}\sqrt{\Delta_{\rm min}}}\log^{3/2}\left(\frac{\sigma}{\epsilon^{\prime\prime}}\right)\log\log\left(\frac{\sigma}{\epsilon^{\prime\prime}}\right)\log\left(\frac{l}{\epsilon\sqrt{\Delta_{\rm min}}}\right)\log\left(\frac{l}{\epsilon}\right)\right)
\end{equation}
because of Lemma \ref{lem:Rtiltil}, and multiplying $l \log\left(\frac{l}{\epsilon}\right)$ to this and replacing $l$ with its upper bound $l_{\rm max}$ yields the bound on the total query number in Eq. (\ref{eq:compQSAApp}).

The statement on $n_{\rm anc}$ is obtained by substituting $\delta^{\prime\prime}$ for $\delta$ and $\epsilon^{\prime\prime}$ for $\epsilon^\prime$ in Eq. (\ref{eq:qubitRtiltil}).

Lastly, $\|\tilde{P}-P\|_{\rm TV}\le\epsilon$ is seen from Lemma \ref{lem:TVvsLerr} and Eq. (\ref{eq:tilLepsprpr}).
\end{proof}

\subsection{Finding the credible interval \label{sec:CI}}

By the above method, we can get an approximation of the quantum state in which the target distribution $P$ is encoded in amplitudes.
However, in practice, our aim is not to get a quantum state but some statistics on $P$.
Although it seems that the previous studies on quantum algorithms for MCMC have not focused on this point, this paper considers it.
Concretely, as a quantity that we typically want, we consider the equal-tailed credible interval of a random variable that obeys $P$.
Formally, writing $x\in\Omega\subset \mathbb{R}^d$ as $x=(x^{(1)},...,x^{(d)})$ and defining $\Omega^{(i)}:=\{x^{(i)} | x\in\Omega\}$, 
we want $x_{\rm ub}^{(i)},x_{\rm lb}^{(i)}\in\Omega^{(i)}$ that satisfy\footnote{In the current setting that $\Omega$ is discrete and so is each $\Omega^{(i)}$, $x_{\rm ub}^{(i)}$ and $x_{\rm lb}^{(i)}$ satisfying Eq. (\ref{eq:defCI}) may not exist. However, for simplicity, we now assume that such $x_{\rm ub}^{(i)}$ and $x_{\rm lb}^{(i)}$ exist. As long as the discretization is sufficiently fine as assumed in Sec. \ref{sec:notation}, it is reasonable to expect that there are $x_{\rm ub}^{(i)},x_{\rm lb}^{(i)}\in\Omega^{(i)}$ such that $P(\{x_i>x_{\rm ub}^{(i)}\})$ and $P(\{x_i<x_{\rm lb}^{(i)}\})$ are much closer to $\frac{\alpha}{2}$ than the accuracy we require.}
\begin{equation}
    P\left(\left\{x_i>x_{\rm ub}^{(i)}\right\}\right)=\frac{\alpha}{2}, P\left(\left\{x_i<x_{\rm lb}^{(i)}\right\}\right)=\frac{\alpha}{2} \label{eq:defCI}
\end{equation}
with a credibility level $\alpha\in(0,1)$ for each $i\in[d]$.
In other words, $x^{(i)}$ is in the interval $\left[x_{\rm lb}^{(i)},x_{\rm ub}^{(i)}\right]$ with probability $\alpha$.
A typical example of this type of problem is parameter estimation by Bayesian inference: with $P$ the posterior distribution of the parameters in some statistical model, we find the bound for each parameter in the above form.

Given the quantum circuit to approximately generate $\ket{P}$, a natural approach is finding $x_{\rm ub}^{(i)}$ by binary search with the CDF $\Phi_{P}(a):=P(\{x^{(i)}>a\})$ computed by QMCI (and $x_{\rm lb}^{(i)}$ is found in the same fashion).
We hereafter elaborate this approach.
First, we describe how to compute $\Phi_{P}(a)$.

\begin{lemma}
    Suppose that Assumptions \ref{ass:P} to \ref{ass:CondNum} are satisfied.
    Then, for any $i\in[d]$, $\epsilon,\delta\in(0,1)$ and $a\in\Omega^{(i)}$, there exists an algorithm that, with probability at least $1-\delta$, outputs an $\epsilon$-approximation of $\Phi_P(a)$, making
    \begin{align}
        &O\left(\frac{\sigma l_{\rm max}}{\epsilon^{\prime\prime}\sqrt{\Delta_{\rm min}}}\log^{3/2}\left(\frac{\sigma}{\epsilon^{\prime\prime}}\right)\log\log\left(\frac{\sigma}{\epsilon^{\prime\prime}}\right)\log\left(\frac{1}{\sqrt{\Delta_{\rm min}}}\right) \right. \nonumber \\
        &\qquad \times 
        \left(\log l_{\rm max}+\log L_{\rm max}\log\left(\frac{l_{\rm max}L_{\rm max}}{\delta}\right)\right) \nonumber \\
        &\quad + \frac{\sigma l_{\rm max}}{\epsilon\epsilon^{\prime\prime} \sqrt{\Delta_{\rm min}}}\log^{3/2}\left(\frac{\sigma}{\epsilon^{\prime\prime}}\right)\log\log\left(\frac{\sigma}{\epsilon^{\prime\prime}}\right)\log\left(\frac{l_{\rm max}}{\epsilon \sqrt{\Delta_{\rm min}}}\right) \nonumber \\
        &\qquad \times \left.\log^2\left(\frac{l_{\rm max}}{\epsilon}\right)\log\left(\frac{1}{\delta}\right) \right) \label{eq:compGetPhia}
    \end{align}      
    uses of $O_{\ell}$. \label{lem:AlgCDF}
\end{lemma}

\begin{proof}
    Because of Theorem \ref{th:QSAbyQMCI}, by Algorithm \ref{alg:QSA}, we get $\beta_0,...,\beta_l$, with which A$\frac{\pi}{3}$AA generates $\widetilde{\ket{\tilde{P}_\star}}$ $\frac{\epsilon}{9}$-close to $\ket{\tilde{P}_\star}\ket{0}^{n_{\rm anc}}$ with $n_{\rm anc}$ of order (\ref{eq:qubitQSAApp}).
    Here, $\tilde{P}_\star$ is a distribution on $\Omega$ such that $\|\tilde{P}_\star-P\|_{\rm TV}\le \frac{\epsilon}{3}$.
    We denote this A$\frac{\pi}{3}$AA by $\mathcal{V}_P$.

    Note that $\widetilde{\ket{\tilde{P}_\star}}$ is written as follows:
    \begin{equation}
        \widetilde{\ket{\tilde{P}_\star}}=\sum_{x\in\Omega} \sqrt{\tilde{P}_\star(x)} \ket{x}\ket{0}^{\otimes n_{\rm anc}} + \hat{\epsilon}\sum_{\hat{x}\in\hat{\Omega}} \sqrt{\hat{P}(\hat{x})} \ket{\hat{x}}\ket{\psi_{\hat{x}}}
    \end{equation}
    where $\hat{\epsilon}\in\left[0,\frac{\epsilon}{9}\right)$, $\hat{\Omega}$ is a finite subset of $\mathbb{R}^d$ that may differ from $\Omega$, $\hat{P}$ is a distribution on $\hat{\Omega}$, and $\ket{\psi_{\hat{x}}}$ is a state on $n_{\rm anc}$ ancillary qubits.
    $\widetilde{\ket{\tilde{P}_\star}}$ can be rewritten as
    \begin{align}
        \widetilde{\ket{\tilde{P}_\star}}&=\sum_{x\in\Omega}  \ket{x}\left(\sqrt{\tilde{P}_\star(x)}\ket{0}^{\otimes n_{\rm anc}}+\hat{\epsilon}\sqrt{\hat{P}(x)}\ket{\psi_{x}}\right) \nonumber \\ 
        & \quad + \hat{\epsilon}\sum_{\hat{x}\in\hat{\Omega}\cap\overline{\Omega}} \sqrt{\hat{P}(\hat{x})} \ket{\hat{x}}\ket{\psi_{\hat{x}}} \nonumber \\ 
        &=\sum_{x\in\Omega} \sqrt{\tilde{P}^{\prime}_\star(x)} \ket{x}\ket{\tilde{\psi}_x} + \hat{\epsilon}\sum_{\hat{x}\in\hat{\Omega}\cap\overline{\Omega}} \sqrt{\hat{P}(\hat{x})} \ket{\hat{x}}\ket{\psi_{\hat{x}}}.
    \end{align}
    Here, $\ket{\tilde{\psi}_{x}}$ is a state on $n_{\rm anc}$ ancillary qubits and
    $\sqrt{\tilde{P}^{\prime}_\star(x)}:=\left\|\sqrt{\tilde{P}_\star(x)}\ket{0}^{\otimes n_{\rm anc}}+\hat{\epsilon}\sqrt{\hat{P}(x)}\ket{\psi_{x}}\right\|$.
    Then, $\sum_{x\in\Omega}\tilde{P}^{\prime}_\star(x)\le 1$ follows from $\left\|\widetilde{\ket{\tilde{P}_\star}}\right\|=1$ and
    \begin{equation}
       \left| \sqrt{\tilde{P}^{\prime}_\star(x)} - \sqrt{\tilde{P}_\star(x)}\right| \le \hat{\epsilon} \sqrt{\hat{P}(x)} \label{eq:temp3}
    \end{equation}
    follows from the triangle inequality.

    On the other hand, we can implement a quantum circuit $O^{\rm comp}_a$ that acts as
    \begin{equation}
        O^{\rm comp}_a\ket{x}\ket{0}=\ket{x}\ket{\mathbf{1}_{x>a}}
    \end{equation}
    using arithmetic circuits.
    Therefore, because of Theorem \ref{th:QMCIBound}, the capability to generate $\widetilde{\ket{\tilde{P}_\star}}$ means that we can get a $\frac{\epsilon}{3}$-approximation $\tilde{\tilde{\Phi}}^\prime(a)$ of
    \begin{equation}
        \tilde{\Phi}^\prime(a):=\sum_{x\in\Omega} \tilde{P}^{\prime}_\star(x) \mathbf{1}_{x>a} + \hat{\epsilon}^2\sum_{\hat{x}\in\hat{\Omega}\cap\overline{\Omega}} \hat{P}(\hat{x})\mathbf{1}_{\hat{x}>a}
    \end{equation}
    by QMCI with probability $1-\delta$.
    Let us see that $\tilde{\tilde{\Phi}}^\prime(a)$ is an $\epsilon$-approximation of $\Phi(a)=\sum_{x\in\Omega} P(x) \mathbf{1}_{x>a}$.
    The difference between $\tilde{\Phi}^\prime(a)$ and $\Phi(a)$ is bounded as
    \begin{align}
        |\tilde{\Phi}^\prime(a)-\Phi(a)| 
        &\le |\tilde{\Phi}^\prime(a)-\tilde{\Phi}(a)| + |\tilde{\Phi}(a)-\Phi(a)|\nonumber \\
        &\le \sum_{x\in\Omega} \left|\tilde{P}^{\prime}_\star(x)-\tilde{P}_\star(x)\right|+\hat{\epsilon}^2\sum_{\hat{x}\in\hat{\Omega}\cap\overline{\Omega}} \hat{P}(\hat{x}) + \frac{\epsilon}{3}. \label{eq:tilPhiPraPhi}
    \end{align}
    Here, $\tilde{\Phi}(a):=\sum_{x\in\Omega} \tilde{P}_\star(x) \mathbf{1}_{x>a}$ and we used $|\tilde{\Phi}(a)-\Phi(a)|\le\frac{\epsilon}{3}$ that follows from $\|\tilde{P}_\star-P\|_{\rm TV}\le \frac{\epsilon}{3}$.
    The first term in Eq. (\ref{eq:tilPhiPraPhi}) is bounded as
    \begin{eqnarray}
        && \sum_{x\in\Omega} \left|\tilde{P}^{\prime}_\star(x)-\tilde{P}_\star(x)\right| \nonumber \\
        &= & \sum_{x\in\Omega} \left| \sqrt{\tilde{P}^{\prime}_\star(x)} - \sqrt{\tilde{P}_\star(x)}\right|\left(\sqrt{\tilde{P}^{\prime}_\star(x)} + \sqrt{\tilde{P}_\star(x)}\right) \nonumber \\
        &\le &\sum_{x\in\Omega}\hat{\epsilon} \sqrt{\hat{P}(x)}\left(\sqrt{\tilde{P}^{\prime}_\star(x)} + \sqrt{\tilde{P}_\star(x)}\right) \nonumber \\
        &\le & \hat{\epsilon}\left(\left(\sum_{x\in\Omega} \hat{P}(x)\right)^{1/2} \left(\sum_{x\in\Omega} \tilde{P}^{\prime}_\star(x)\right)^{1/2} \right. \nonumber \\
        &&\left.\quad +\left(\sum_{x\in\Omega} \hat{P}(x)\right)^{1/2} \left(\sum_{x\in\Omega} \tilde{P}_\star(x)\right)^{1/2}\right) \nonumber \\
        &\le & 2 \hat{\epsilon} \nonumber \\
        &\le & \frac{2}{9}\epsilon,
    \end{eqnarray}
    where we use Eq. (\ref{eq:temp3}) at the first inequality and the Cauchy–Schwarz inequality at the second inequality.
    The second term in Eq. (\ref{eq:tilPhiPraPhi}) is bounded as
    \begin{equation}
        \hat{\epsilon}^2\sum_{\hat{x}\in\hat{\Omega}\cap\overline{\Omega}} \hat{P}(\hat{x}) \le \hat{\epsilon}^2\le\hat{\epsilon}\le \frac{\epsilon}{9}.
    \end{equation}
    Consequently, $\tilde{\Phi}(a)$ is a $\frac{2}{3}\epsilon$-approximation of $\Phi(a)$, which means that $\tilde{\tilde{\Phi}}^\prime(a)$ is an $\epsilon$-approximation of $\Phi(a)$.

    Finally, let us check the query complexity.
    To get $\tilde{\beta}_0,...,\tilde{\beta}_l$ with probability at least $1-\frac{\delta}{2}$ by Algorithm \ref{alg:QSA}, we make a number of order (\ref{eq:compGetBetaApp}) of calls to $O_{\ell}$.
    After this, to get $\tilde{\tilde{\Phi}}^\prime(a)$ with probability at least $1-\frac{\delta}{2}$ by QMCI, we call the circuit to generate $\widetilde{\ket{\tilde{P}_\star}}$ $O\left(\frac{1}{\epsilon}\log\delta^{-1}\right)$ times, and one call to this circuit contains a number of order (\ref{eq:compQSAApp}) of calls to $O_{\ell}$.
    Summing up these, we see that the total number of calls to $O_{\ell}$ is of order (\ref{eq:compGetPhia}).
\end{proof}

Then, we reach the algorithm to find $x_{\rm ub}^{(i)}$.

\begin{theorem}
    Suppose that Assumptions \ref{ass:P} to \ref{ass:CondNum} are satisfied.
    Let $i\in[d]$, $\alpha\in(0,1)$, $\delta\in(0,1)$ and $\epsilon\in\left(0,\frac{\alpha}{2}\right)$. 
    Suppose that there exists $x^{(i)}\in\Omega^{(i)}$ such that
    \begin{equation}
       \frac{\alpha}{2} - \frac{\epsilon}{3} \le \Phi(x^{(i)}) \le \frac{\alpha}{2} + \frac{\epsilon}{3}.
       \label{eq:xubtilAssum}
    \end{equation}
    Then, there exists an algorithm that, with probability at least $1-\delta$, outputs $\tilde{x}_{\rm ub}^{(i)}\in\Omega^{(i)}$ such that
    \begin{equation}
       \frac{\alpha}{2} - \epsilon \le \Phi(\tilde{x}_{\rm ub}^{(i)}) \le \frac{\alpha}{2} + \epsilon.
       \label{eq:xubtil}
    \end{equation}
    making
    \begin{align}
        &O\left(\frac{\sigma l_{\rm max}}{\epsilon^{\prime\prime}\sqrt{\Delta_{\rm min}}}\log^{3/2}\left(\frac{\sigma}{\epsilon^{\prime\prime}}\right)\log\log\left(\frac{\sigma}{\epsilon^{\prime\prime}}\right)\log\left(\frac{1}{\sqrt{\Delta_{\rm min}}}\right) \right.\nonumber \\
        &\qquad\quad \times\left(\log l_{\rm max}+\log L_{\rm max}\log\left(\frac{l_{\rm max}L_{\rm max}}{\delta}\right)\right)  \nonumber \\
        &\qquad +\frac{\sigma l_{\rm max}}{\epsilon\epsilon^{\prime\prime} \sqrt{\Delta_{\rm min}}}\log^{3/2}\left(\frac{\sigma}{\epsilon^{\prime\prime}}\right)\log\log\left(\frac{\sigma}{\epsilon^{\prime\prime}}\right)\log\left(\frac{l_{\rm max}}{\epsilon \sqrt{\Delta_{\rm min}}}\right) \nonumber \\
        &\qquad\quad \left.\times\log^2\left(\frac{l_{\rm max}}{\epsilon}\right)\log\left(\frac{1}{\delta}\right)\log\left(\frac{\log|\Omega^{(i)}|}{\delta}\right) \log|\Omega^{(i)}|\right). \label{eq:compGetCI}
    \end{align}
    queries to $O_{\ell}$.
    \label{th:AlgCI}
\end{theorem}

For $\epsilon$ satisfying Eq. (\ref{eq:smallEpsCond}), Eq. (\ref{eq:compGetCI}) becomes
\begin{equation}
    \tilde{O}\left(\frac{\sigma \bar{L}^{1/2}}{\epsilon^2 \Delta_{\rm min}^{3/2}}\right).
    \label{eq:compGetCIWoLog}
\end{equation}

\begin{proof}[Proof of Theorem \ref{th:AlgCI}]
    The algorithm is presented as Algorithm \ref{alg:CI}.

\begin{algorithm}[H]
    \begin{algorithmic}[1]
    \REQUIRE {
        \ \\
        \begin{itemize}        




        \item Accuracy $\epsilon\in(0,1)$

        \item Failure probability $\delta\in(0,1)$

        \item Credibility level $\alpha\in(0,1)$
        
        \end{itemize}
    }

    \STATE By Algorithm \ref{alg:QSA}, get $\tilde{\beta}_0,...,\tilde{\beta}_l$, with which A$\frac{\pi}{3}$AA generates a state $\frac{\epsilon}{27}$-close to $\ket{\tilde{P}_\star}\ket{0}^{\otimes n_{\rm anc}}$, where $n_{\rm anc}$ is of order (\ref{eq:qubitQSAApp}) and $\tilde{P}_\star$ is a distribution on $\Omega$ satisfying $\|\tilde{P}_\star-P\|_{\rm TV}\le\frac{\epsilon}{9}$.
    
    \STATE Using the obtained $\tilde{\beta}_0,...,\tilde{\beta}_l$, construct the above A$\frac{\pi}{3}$AA as a quantum circuit $\mathcal{V}_P$.

    \STATE Sort the elements of $\Omega^{(i)}$ in the ascending order and denote them by $x^{(i)}_1,...,x^{(i)}_{n_i}$, where $n_i:=|\Omega^{(i)}|$.

    \STATE By QMCI with $\mathcal{V}_P$, get a $\frac{\epsilon}{3}$-approximation $\tilde{\tilde{\Phi}}^\prime(x^{(i)}_1)$ of $\Phi(x^{(i)}_1)$ with failure probability $\delta^\prime:=\frac{\delta}{n_{\rm max}+1}$, where $n_{\rm max}:=\left\lceil\log_2 (n_i-2)\right\rceil+1$ (we do not need to compute $\Phi(x^{(i)}_{n_i})$ since it is 0).

    \IF{$\left|\tilde{\tilde{\Phi}}^\prime(x^{(i)}_1)-\frac{\alpha}{2}\right|\le \frac{2}{3}\epsilon$}
        \STATE Output $x^{(i)}_1$ as $\tilde{x}_{\rm ub}^{(i)}$ and stop.
    \ELSIF{$\tilde{\tilde{\Phi}}^\prime(x^{(i)}_1)<\frac{\alpha}{2}-\frac{2}{3}\epsilon$}
        \STATE Stop with no output.
    \ENDIF

    \STATE Set $j_{\rm ub}=n_i$ and $j_{\rm lb}=1$.

    \REPEAT 

    \STATE Set $j_{\rm mid}=\left\lceil\frac{j_{\rm ub}+j_{\rm lb}}{2}\right\rceil$.

    \STATE By QMCI with $\tilde{V}_P$, get a $\frac{\epsilon}{3}$-approximation $\tilde{\tilde{\Phi}}^\prime(x^{(i)}_{j_{\rm mid}})$ of $\Phi(x^{(i)}_{j_{\rm mid}})$ with failure probability $\delta^\prime$.

    \IF{$\left|\tilde{\tilde{\Phi}}^\prime(x^{(i)}_{j_{\rm mid}})-\frac{\alpha}{2}\right|\le \frac{2}{3}\epsilon$}
        \STATE Output $x^{(i)}_{j_{\rm mid}}$ as $\tilde{x}_{\rm ub}^{(i)}$ and stop.
    \ELSIF{$\tilde{\tilde{\Phi}}^\prime(x^{(i)}_{j_{\rm mid}})>\frac{\alpha}{2}+\frac{2}{3}\epsilon$}
        \STATE Set $j_{\rm lb}=j_{\rm mid}$.
    \ELSE [in this case, $\tilde{\tilde{\Phi}}^\prime(x^{(i)}_{j_{\rm mid}})<\frac{\alpha}{2}-\frac{2}{3}\epsilon$]
        \STATE Set $j_{\rm ub}=j_{\rm mid}$.
    \ENDIF

    \UNTIL{$j_{\rm ub}-j_{\rm lb}=1$}

    \STATE Output $x^{(i)}_{j_{\rm mid}}$.
    
    \caption{Algorithm to find $x_{\rm ub}^{(i)}$}
    \label{alg:CI}
    \end{algorithmic}
\end{algorithm}

Then, let us show that this algorithm has a property stated in the theorem.

First, note that the loop in lines 11 to 21 ends in at most $n_{\rm max}$ iterations.
To see this, denoting $j_{\rm ub}$ and $j_{\rm lb}$ at the end of the $k$th iteration by $j_{{\rm ub},k}$ and $j_{{\rm lb},k}$, respectively, we notice that
\begin{equation}
    j_{{\rm ub},k+1}-j_{{\rm lb},k+1} \le \frac{j_{{\rm ub},k}-j_{{\rm lb},k}}{2} + \frac{1}{2},
\end{equation}
which implies
\begin{equation}
    j_{{\rm ub},k}-j_{{\rm lb},k} \le 2^{-k} (n_i-2)+1.
\end{equation}
Thus, $j_{{\rm ub},k}-j_{{\rm lb},k}$ becomes 2 or less in at least $\left\lceil\log_2 (n_i-2)\right\rceil$ iterations, and, even if it becomes 2, the next iteration makes it 1.
Therefore, the loop ends in $n_{\rm max}$ iterations by the condition $j_{\rm ub}-j_{\rm lb}=1$, or earlier by the condition $\left|\tilde{\tilde{\Phi}}^\prime(x^{(i)}_{j_{\rm mid}})-\frac{\alpha}{2}\right|\le \frac{2}{3}\epsilon$.

Let us consider the case that all the QMCIs in the algorithm, that in line 4 and those in loop 11-21, successfully outputs $\frac{\epsilon}{3}$-approximations of $\Phi(x^{(i)}_1)$ and $\Phi(x^{(i)}_{j_{\rm mid}})$.
This occurs with probability at least $\left(1-\delta^\prime\right)^{n_{\rm max}+1}\ge \delta$.
In these QMCIs, if we obtain $\left|\tilde{\tilde{\Phi}}^\prime(x^{(i)})-\frac{\alpha}{2}\right|\le \frac{2}{3}\epsilon$ for some $x^{(i)}$, $\left|\Phi(x^{(i)})-\frac{\alpha}{2}\right|\le \epsilon$ also holds because of $\left|\tilde{\tilde{\Phi}}^\prime(x^{(i)})-\Phi(x^{(i)})\right|\le \frac{\epsilon}{3}$.
In fact, we get such $x^{(i)}$ with certainty under the condition that all the QMCIs succeed.
This is seen by contradiction.
Suppose that, under this condition, loop 11-21 ends with $j_{\rm ub}-j_{\rm lb}=1$.
This means that $\tilde{\tilde{\Phi}}^\prime(x^{(i)}_{j_{\rm lb}})>\frac{\alpha}{2}+\frac{2}{3}\epsilon$ and $\tilde{\tilde{\Phi}}^\prime(x^{(i)}_{j_{\rm ub}})<\frac{\alpha}{2}-\frac{2}{3}\epsilon$, which leads to
\begin{equation}
\Phi(x^{(i)}_{j_{\rm lb}})>\frac{\alpha}{2}+\frac{\epsilon}{3}, \Phi(x^{(i)}_{j_{\rm ub}})<\frac{\alpha}{2}-\frac{\epsilon}{3}.
\label{eq:contradict}
\end{equation}
Since $\Phi$ is monotonically decreasing and there is no $\Omega^{(i)}$'s element between $x^{(i)}_{j_{\rm lb}}$ and $x^{(i)}_{j_{\rm ub}}$, Eq. (\ref{eq:contradict}) contradicts with the assumption that Eq. (\ref{eq:xubtilAssum}) holds for some $x^{(i)}\in\Omega^{(i)}$.

In summary, with probability at least $1-\delta$, $\tilde{x}^{(i)}_{\rm ub}$ satisfying Eq. (\ref{eq:xubtil}) is output after either of QMCIs.

The statement on the query complexity immediately follows from Lemma \ref{lem:AlgCDF}.
The first term in Eq. (\ref{eq:compGetCI}) correspond to finding $\tilde{\beta}_0,...,\tilde{\beta}_l$ and is similar to the first term in Eq. (\ref{eq:compGetPhia}).
The second term in Eq. (\ref{eq:compGetCI}) corresponds to QMCIs and is obtained by multiplying the number of QMCIs, which is of order $O\left(\log |\Omega^{(i)}|\right)$, to the second term in Eq. (\ref{eq:compGetPhia}), and substituting $\delta^\prime$ for $\delta$.

\end{proof}


Seemingly, the statement in Theorem \ref{th:AlgCI} is tricky: it assumes the existence of $x^{(i)}$ for which $\Phi(x^{(i)})$ is $\frac{\epsilon}{3}$-close to $\frac{\alpha}{2}$, but only guarantees that the algorithm's output is $\epsilon$-close.
This is because of the erroneous nature of QMCI.
Suppose that we search $x^{(i)}$ such that $\left|\Phi(x^{(i)})-\frac{\alpha}{2}\right|\le\epsilon$ and there exists $x^{(i)}$ that marginally satisfies this.
Then, even if we require high accuracy in QMCI, it may output an estimate of $\Phi(x^{(i)})$ out of the $\epsilon$-neighborhood of $\frac{\alpha}{2}$, which makes us fail to notice that $x^{(i)}$ is what we want.
We thus conduct QMCIs with accuracy $\frac{\epsilon}{3}$ and pick up $x^{(i)}$ with $\tilde{\tilde{\Phi}}^\prime(x^{(i)})$ $\frac{2}{3}\epsilon$-close to $\frac{\alpha}{2}$ as an answer.
Under this policy, we never miss $x^{(i)}$ satisfying Eq. (\ref{eq:xubtilAssum}), since the $\frac{\epsilon}{3}$-approximation of $\Phi(x^{(i)})$ is never out of the $\frac{2}{3}\epsilon$-neighborhood of $\frac{\alpha}{2}$.
Of course, we might pick up $x^{(i)}$ for which $\left|\Phi(x^{(i)})-\frac{\alpha}{2}\right|>\frac{2}{3}\epsilon$, given the erroneous QMCI estimate of $\Phi(x^{(i)})$ accidentally lying in the $\frac{2}{3}\epsilon$-neighborhood of $\frac{\alpha}{2}$.
Even if so, the chosen $x^{(i)}$ at least satisfies $\left|\Phi(x^{(i)})-\frac{\alpha}{2}\right|\le\epsilon$, since the $\frac{\epsilon}{3}$-approximation of a number distant from $\frac{\alpha}{2}$ by more than $\epsilon$ never lies in the $\frac{2}{3}\epsilon$-neighborhood of $\frac{\alpha}{2}$.

A similar discussion is found in consideration on setting the threshold of the SNR in the quantum algorithm for GW matched filtering proposed in \cite{Miyamoto2022}. 

\subsection{Comparison with other approaches \label{sec:CIComp}}

We now make a comparison of the above method for finding the credible interval with other approaches.
We compare the order of the number of queries to $O_{\ell}$ in the various approaches except logarithmic factors.
Since the binary search adds only logarithmic factors, it is sufficient to consider the complexity of calculating the CDF within accuracy $\epsilon$.

First, let us consider QSA {\it without} QMCI.
That is, the state $\ket{P}$ that encodes the target distribution $P$ is prepared via Algorithm \ref{alg:QSA} and A$\frac{\pi}{3}$AA with the obtained $\{\beta_i\}$, with $L$ computed by not QMCI but $M$-time iterated calculations and additions of $\ell$.
Then, using this state-preparing circuit as $\mathcal{V}_P$, we estimate $x^{(i)}_{\rm ub}$ by Algorithm \ref{alg:CI}.
We call this the exact QSA approach.
Note that, the quantum walk operator $U$, which is now exact one in Eq. (\ref{eq:walkOp}), makes $O(M)$ calls to $O_{\ell}$.
Combining this with Eq. (\ref{eq:compQSA}), we see that the number of calls to $O_{\ell}$ in generating $\ket{P}$ by QSA is
\begin{equation}
    \tilde{O}\left(\frac{M\bar{L}^{1/2}}{\Delta_{\rm min}^{1/2}}\right).
\end{equation}
Besides, considering the complexity of QMCI in Eq. (\ref{eq:compQMCIBound}), we estimate the total number of calls to $O_{\ell}$ in finding a credible interval in the exact QSA approach as
\begin{equation}
    \tilde{O}\left(\frac{M\bar{L}^{1/2}}{\Delta_{\rm min}^{1/2}\epsilon}\right). \label{eq:compexactQSA}
\end{equation}

Next, let us consider the fully classical approach: on a classical computer, generating the Markov chain by the MH method in Algorithm \ref{alg:MH}, with $L$ obtained by $M$-time iterative calculations.
We now regard $O_{\ell}$ as a classical subroutine to compute $\ell$.
Based on the bound (\ref{eq:MCSampEps2}) on the step number in MCMC-based expectation estimation, the total number of calls to $O_{\ell}$ in finding the credible interval is
\begin{equation}
    \tilde{O}\left(\frac{M}{\Delta\epsilon^2}\right). \label{eq:compFullCl}
\end{equation}
Also note that we do not need binary search with respect to the CDF in the classical approach.
We can store the sampled states on a classical memory, and thus, sorting them and taking the $100\left(1-\frac{\alpha}{2}\right)$th percentile yields an estimate on $x^{(i)}_{\rm ub}$. 


Note that the complexity of the proposed method in Eq. (\ref{eq:compGetCIWoLog}) is not better than those of the exact QSA approach and the fully classical approach in Eqs. (\ref{eq:compexactQSA}) and (\ref{eq:compFullCl}) with respect to the spectral gap and accuracy.
On the other hand, unlike Eqs. (\ref{eq:compexactQSA}) and (\ref{eq:compFullCl}), the complexity of the proposed method is not explicitly dependent on $M$ the number of terms in $L$ but on $\sigma$ the standard deviation of $\ell$.
Thus, the proposed method can be advantageous with respect to $M$ if $\sigma$ scales with $M$ sublinearly, and this actually holds in the case of GW parameter estimation considered in Sec. \ref{sec:GW}.


\section{Application: parameter estimation in gravitational wave detection experiments \label{sec:GW}}

As an application of the credible interval calculation method proposed above, we consider parameter estimation in GW experiments.
Since the first detection in 2015 \cite{GW150914}, GW events have been detected by laser interferometers such as LIGO and Virgo \cite{GWTC-1,GWTC-2,GWTC-3}.
Given a GW event, we want to estimate the parameters of the GW, such as masses of the sources for a GW from a compact binary coalescence (CBC).
For this purpose, Bayesian inference with MCMC is widely used (for a review, see \cite{thrane_talbot_2019}).
Given the detector output $s(t)$ as time-series data with time length $T$ and interval $\Delta t$, the negative log-likelihood for a point $x$ in the parameter space is given as follows:
\begin{align}
    &L(x)=-2\Re\left(h(\cdot,x) | s\right)+\left(h(\cdot,x) | h(\cdot,x)\right)  + C\nonumber \\
    &\Re\left(h(\cdot,x) | s\right)=\frac{4}{M}\sum_{k=1}^{\frac{M}{2}-1}\Re\left(\frac{\tilde{h}^*(f_k;x)\tilde{s}(f_k)}{S_{\rm n}(f_k)\Delta t}\right) \nonumber \\
    &\left(h(\cdot,x) | h(\cdot,x)\right) = \frac{4}{M}\sum_{k=1}^{\frac{M}{2}-1}\frac{|\tilde{h}(f_k;x)|^2}{S_{\rm n}(f_k)\Delta t}. \label{eq:LGW}
\end{align}
Here, $M=\frac{T}{\Delta t}$, $f_k:=\frac{k}{T}$, the tilde represents the Fourier transform of a function of time, $h(t,x)$ is the GW waveform for $x$, $S_{\rm n}$ is the single-sided power spectrum density of the noise, and $C$ is a term independent of $x$.
Since $\tilde{h}(\cdot,x)$ and $S_{\rm n}$ are smooth functions evaluated by explicit formulas, we assume that $\left(h(\cdot,x) | h(\cdot,x)\right)$ is approximated by the integral $4\int_0^\infty \frac{|\tilde{h}(f;x)|^2}{S_{\rm n}(f)}df$ and this is further approximated by some formula efficiently computable by arithmetic circuits.
Then, $L$ in Eq. (\ref{eq:LGW}) is in the form of Eq. (\ref{eq:likeli}).
In fact, $M$ can be as large as $10^6-10^{10}$ in typical cases \cite{Miyamoto2022}, and thus we are motivated to apply our QMCI-based method in Sec. \ref{sec:ourAlg} to find credible intervals for GW parameters, regarding $-2\Re\left(h(\cdot,x) | s(t)\right)$ as $L_{\rm sum}$ and $-4\Re\left(\frac{\tilde{h}^*(f_k;x)\tilde{s}(f_k)}{S_{\rm n}(f_k)\Delta t}\right)$ as $\ell(k,x)$.

Note that other conditions to apply the proposed method are met.
Usually, we have found a high SNR point in the parameter space by matched filtering conducted prior to parameter estimation, and thus we can set a parameter region to be searched, for example a hyperrectangle around such a point.
We can set $\Omega$ to the sufficiently dense discrete points in that region.
Commonly, the prior distribution $P_0$ is set to uniform on $\Omega$ and the proposal distribution $T(x,\cdot)$ is set to some easy-to-sample one such as the normal distribution around $x$, which means Assumptions \ref{ass:oraclesForT}, \ref{ass:OP0} and \ref{ass:OAR} are satisfied.
On the other hand, since the detector output is affected by the random noise and unable to be expressed as an analytic formula, $O_{\ell}$ is not implemented as a combination of arithmetic circuits.
Nevertheless, if we assume the availability of quantum random access memory (QRAM) \cite{Giovannetti2008QRAM}, we can implement $O_{\ell}$ using a QRAM that stores the values of $\tilde{s}(f_k)$, and thus Assumption \ref{ass:P} is satisfied.
The preparation of such a QRAM takes $O(M)$ time, but this is needed only once at the very beginning of calculation.

Let us estimate the query complexity of credible interval calculation for GW parameters by the proposed method.
To do so, we need to bound the variance $\sigma^2$ of terms in $\Re\left(h(\cdot,x) | s(t)\right)$.
According to \cite{Miyamoto2022}, $\sigma=O\left(\gamma M^{1/2}\right)$  with
\begin{equation}
    \gamma:=\max_{\substack{x\in\Omega \\ k\in\left[\frac{M}{2}-1\right]}}\frac{\tilde{h}(f_k,x)}{\sqrt{S_{\rm n}(f_k)\Delta t}},
\end{equation}
which is $O(1)$ in some cases if $h$ is normalized so that $\left(h(\cdot,x) | h(\cdot,x)\right)=1$ as in matched filtering \cite{Miyamoto2022}.
If we get the highest SNR $\rho=\left(h(\cdot,x_\star) | s\right)$ at the parameter point $x_\star$ with $\left(h(\cdot,x_\star) | h(\cdot,x_\star)\right)=1$ in matched filtering, $\Omega$ should be set in the neighbor of the parameter that corresponds to the waveform $\rho h(\cdot,x_\star)$, which leads to $\gamma=O(\rho)$ and $\sigma=O\left(\rho M^{1/2}\right)$.
Then, Eq. (\ref{eq:compGetCIWoLog}) becomes
\begin{equation}
    \tilde{O}\left(\frac{M^{1/2} \rho \bar{L}}{\epsilon^2 \Delta_{\rm min}^{3/2}}\right).
\end{equation}
Compared to the complexity of the exact QSA approach in Eq. (\ref{eq:compexactQSA}), the proposed method provides the quadratic speedup with respect to $M$, in compensation for the worse scaling on $\epsilon$ and $\Delta_{\rm min}$.

\section{Summary \label{sec:summary}}

In this paper, with the usage in Bayesian inference in mind, we have considered the quantum version of the MH algorithm in the case that the target probability $P$ is in the form of Eq. (\ref{eq:PeL}) and $L$ is given as Eqs. (\ref{eq:likeli}) and (\ref{eq:Lsum}) with large $M$, based on QSA.
In such a case, calculating $L$ takes the $O(M)$ query complexity naively, and thus we have proposed application of QMCI, which may speedup a costly summation.
We have presented not only the procedure to generate the state that encodes $P$ but also that for finding a credible interval of a parameter in a statistical model.
Setting the accuracy in QMCI based on the result in \cite{alquier2016noisy} on the MH algorithm with the perturbed acceptance ratio, we have derived the bound on the complexity, the number of calls to the quantum circuit to compute $\ell$, as summarized in Table \ref{tbl:CompSum}.
Comparing QSA with $L$ calculated exactly, the complexity of the proposed method scales worse on the required accuracy $\epsilon$ and the spectral gap $\Delta_{\rm min}$.
On the other hand, if $\sigma$ the standard deviation of $\ell$ scales on $M$ sublinearly, the proposed method is advantageous with respect to $M$.
As an example in which this holds, we have considered estimation of GW parameters in a GW detection experiment.
In this example, $\sigma$ scales on $M$ as $O\left(\sqrt{M}\right)$ and this results in the complexity shown in Table \ref{tbl:CompSum}, which is quadratically smaller with respect to $M$ compared to the exact QSA method.

\section*{acknowledgement}

This work is supported by MEXT Quantum Leap Flagship Program (MEXT Q-LEAP) Grant no. JPMXS0120319794 and JSPS
KAKENHI Grant no. JP22K11924.

\appendix

\section{Properties of the walk operator \label{sec:proofWalkOp}}

Here, we present the proof of Theorem \ref{th:phasegapOurs} on the spectrum of the quantum walk operator.
That of Theorem \ref{th:phasegap} is almost same with $S$ seen as $I$. 

First, we present the following theorem on which our proof is based.

\begin{theorem}[Theorem 1 in \cite{Szegedy2004}]

    Let $\mathcal{H}$ be a $N$-dimensional Hilbert space.
    Let $\mathcal{A}$ (resp. $\mathcal{B})$ be a $n$-dimensional subspace of $\mathcal{H}$ spanned by orthonormal vectors $u_1,...,u_m$ (resp. $v_1,...,v_n$).
    Denote by $V_{\mathcal{A}}$ (resp. $V_{\mathcal{B}}$) the $N \times m$ (resp. $N \times n$) matrix whose $i$th column is $u_i$ (resp. $v_i$).
    Define $R_{\mathcal{A}}=2V_{\mathcal{A}}V_{\mathcal{A}}^\dagger-I$ and $R_{\mathcal{B}}=2V_{\mathcal{B}}V_{\mathcal{B}}^\dagger-I$.
    Then, on $\mathcal{A}+\mathcal{B}$, the unitary operator $R_{\mathcal{A}}R_{\mathcal{B}}$ has an eigenvalue 1 with multiplicity 1, and any other eigenvalue is either of $e^{2i\theta_1},e^{-2i \theta_1},...,e^{2i\theta_l},e^{-2i\theta_l}$ or -1, where $\theta_1,...,\theta_l\in\left(0,\frac{\pi}{2}\right)$ are written as $\theta_i=\arccos \lambda_i$ with singular values $\{\lambda_i\}$ of $V_{\mathcal{A}}^\dagger V_{\mathcal{B}}$ that lie in $(0,1)$.
    
    \label{th:Szegedy}
\end{theorem}

\noindent In the current case, $\mathcal{A}$ and $\mathcal{B}$ are defined as Eq. (\ref{eq:subspaceB}).

We also use the following lemmas.

\begin{lemma}

On $\mathcal{A}$, $\Pi_0 V^\dagger B^\dagger SFBV \Pi_0$ has the same eigenvalues as $W$ including multiplicity.
    \label{lem:SameEigW}
\end{lemma}

\begin{proof}
    By a straightforward calculation, we see that, for any $x\in\Omega$, applying $\Pi_0 V^\dagger B^\dagger SFBV$ to $\ket{x}_{R_{\rm S}}\ket{0}_{R_{\rm M}}\ket{0}_{R_{\rm C}}$ yields
    \begin{widetext}
    \begin{align}
        &\left[\sum_{\Delta x\in\Omega_x\setminus\{\vec{0}_d\}} \sqrt{T(x,x+\Delta x)T(x+\Delta x,x)A(x,x+\Delta x)A(x+\Delta x,x)}\ket{x+\Delta x}_{R_{\rm S}}\right. \nonumber \\
        &\left. \quad + \left(1-\sum_{\Delta x\in\Omega_x\setminus\{\vec{0}_d\}} T(x,x+\Delta x)A(x,x+\Delta x)\right)\ket{x}_{R_{\rm S}}\right]\ket{0}_{R_{\rm M}}\ket{0}_{R_{\rm C}} \nonumber \\
        =& \sum_{y\in\Omega} \sqrt{W_{x,y}W_{y,x}}\ket{y}_{R_{\rm S}}\ket{0}_{R_{\rm M}}\ket{0}_{R_{\rm C}},
    \end{align}
    \end{widetext}
    where $\vec{0}_d$ is the $d$-dimensional zero vector.
    This means that
    \begin{align}
        &\Pi_0 V^\dagger B^\dagger SFBV \Pi_0 \nonumber \\
        =& \sum_{x,y\in\Omega} \sqrt{W_{x,y}W_{y,x}}\ket{y}_{R_{\rm S}}\ket{0}_{R_{\rm M}}\ket{0}_{R_{\rm C}}\bra{x}_{R_{\rm S}}\bra{0}_{R_{\rm M}}\bra{0}_{R_{\rm C}}.
    \end{align}
    Using the detailed balance condition $P(x)W_{x,y}=P(y)W_{y,x}$, which is satisfied in the MH algorithm \cite[EXERCISE 3.1]{levin2017markov}, we have
    \begin{align}
        &\Pi_0 V^\dagger B^\dagger SFBV \Pi_0 \nonumber \\
        =& \sum_{x,y\in\Omega} \sqrt{\frac{P(x)}{P(y)}}W_{x,y}\ket{y}_{R_{\rm S}}\ket{0}_{R_{\rm M}}\ket{0}_{R_{\rm C}}\bra{x}_{R_{\rm S}}\bra{0}_{R_{\rm M}}\bra{0}_{R_{\rm C}} \nonumber \\
        =& \sum_{x,y\in\Omega} (D_P W D_P^{-1})_{x,y}\ket{y}_{R_{\rm S}}\ket{0}_{R_{\rm M}}\ket{0}_{R_{\rm C}}\bra{x}_{R_{\rm S}}\bra{0}_{R_{\rm M}}\bra{0}_{R_{\rm C}}, \label{eq:DPWDP}
    \end{align}
    where $D_P$ is a diagonal matrix indexed by $x,y\in\Omega$ and its $(x,x)$ entry is $\sqrt{P(x)}$.
    Thus, since $\Pi_0 V^\dagger B^\dagger SFBV \Pi_0$ is expressed as conjugation of $W$ by $D_P$, it has the same eigenvalues as $W$ on $\mathcal{A}$. 
\end{proof}

\begin{lemma}

$\ket{P}$ is the eigenstate of $U_W=R V^\dagger B^\dagger SFBV$ with eigenvalue 1. \label{lem:Peig}
    
\end{lemma}

\begin{proof}
    This is shown by a straightforward calculation.
    Applying $FBV$ to $\ket{P}=\sum_{x\in\Omega}\sqrt{P(x)}\ket{x}_{R_{\rm S}}\ket{0}_{R_{\rm M}}\ket{0}_{R_{\rm C}}$ yields
    \begin{widetext}
       \begin{align}
        &\sum_{x\in\Omega}\sum_{\Delta x\in\Omega_x} \left(\sqrt{P(x)T(x,x+\Delta x)A(x,x+\Delta x)}\ket{x+\Delta x}_{R_{\rm S}}\ket{\Delta x}_{R_{\rm M}}\ket{1}_{R_{\rm C}}\right. \nonumber \\
        &\qquad\qquad\quad \left.+\sqrt{P(x)T(x,x+\Delta x)\left(1-A(x,x+\Delta x)\right)}\ket{x}_{R_{\rm S}}\ket{\Delta x}_{R_{\rm M}}\ket{0}_{R_{\rm C}}\right). \label{eq:tocyuu1}
    \end{align} 
    \end{widetext}
    By using the detailed balance condition
    \begin{align}
    & P(x)T(x,x+\Delta x)A(x,x+\Delta x)= \nonumber \\
    & \quad P(x+\Delta x)T(x+\Delta x,x)A(x+\Delta x,x)
    \end{align}
    and substituting $\Delta x$ and $x+\Delta x$ with $-\Delta x$ and $x$, respectively, in the first term in Eq. (\ref{eq:tocyuu1}), we get
    \begin{widetext}
   \begin{align}
        &\sum_{x\in\Omega}\sum_{\Delta x\in\Omega_x} \left(\sqrt{P(x)T(x,x+\Delta x)A(x,x+\Delta x)}\ket{x}_{R_{\rm S}}\ket{-\Delta x}_{R_{\rm M}}\ket{1}_{R_{\rm C}}\right. \nonumber \\
        & \qquad\qquad\qquad\left.+\sqrt{P(x)T(x,x+\Delta x)\left(1-A(x,x+\Delta x)\right)}\ket{x}_{R_{\rm S}}\ket{\Delta x}_{R_{\rm M}}\ket{0}_{R_{\rm C}}\right),
    \end{align}
    \end{widetext}
    and, by applying $S$ to this, we obtain
   \begin{align}
        &\sum_{x\in\Omega}\sum_{\Delta x\in\Omega_x} \sqrt{P(x)T(x,x+\Delta x)}\ket{x}_{R_{\rm S}}\ket{\Delta x}_{R_{\rm M}} \nonumber \\
        &
         \otimes\left(\sqrt{A(x,x+\Delta x)}\ket{1}_{R_{\rm C}}+\sqrt{\left(1-A(x,x+\Delta x)\right)}\ket{0}_{R_{\rm C}}\right).
    \end{align}
    Thus, applying $V^\dagger B^\dagger$ to this yields $\ket{P}$.
    Applying $R$ at last does not change $\ket{P}$.
    
\end{proof}

\begin{lemma}
    The restriction of $\Pi_0 V^\dagger B^\dagger SFBV \Pi_0$ to $\mathcal{A}$ is equal to $V_\mathcal{A}^\dagger V_\mathcal{B}$.
    \label{lem:AdagB}
\end{lemma}

\begin{proof}
    Label the elements in $\Omega$ with integers $1,...,|\Omega|$ and denote the $k$th element by $x_k$.
    Then, for $k,l\in[|\Omega|]$, the $(k,l)$ entry of the restriction of $\Pi_0 V^\dagger B^\dagger SFBV \Pi_0$ to $\mathcal{A}$ is
    \begin{align}
        &\bra{x_k}_{R_{\rm S}}\bra{0}_{R_{\rm M}}\bra{0}_{R_{\rm C}}\Pi_0 V^\dagger B^\dagger SFBV \Pi_0\ket{x_l}_{R_{\rm S}}\ket{0}_{R_{\rm M}}\ket{0}_{R_{\rm C}} \nonumber \\
        =& \bra{x_k}_{R_{\rm S}}\bra{0}_{R_{\rm M}}\bra{0}_{R_{\rm C}} V^\dagger B^\dagger SFBV \ket{x_l}_{R_{\rm S}}\ket{0}_{R_{\rm M}}\ket{0}_{R_{\rm C}}.
    \end{align}
    From the definitions of $V_\mathcal{A}$ and $V_\mathcal{B}$, we see that this is also the $(k,l)$ entry of $V_\mathcal{A}^\dagger V_\mathcal{B}$. 
    
\end{proof}

\begin{lemma}
    On $\mathcal{A}+\mathcal{B}$
    \begin{equation}
        (SF)^\dagger = SF.
    \end{equation}
    \label{lem:SF}
\end{lemma}

\begin{proof}
    For any $x\in\Omega$ and $\Delta x\in\Omega_x$,
    \begin{equation}
        SFSF \ket{x}_{R_{\rm S}}\ket{\Delta x}_{R_{\rm M}}\ket{0}_{R_{\rm C}} = \ket{x}_{R_{\rm S}}\ket{\Delta x}_{R_{\rm M}}\ket{0}_{R_{\rm C}},
    \end{equation}
    since both $S$ and $F$ are not activated if the state on $R_{\rm C}$ is $\ket{0}_{R_{\rm C}}$.
    Besides, applying $F$, $S$, $F$ and $S$ to $\ket{x}_{R_{\rm S}}\ket{\Delta x}_{R_{\rm M}}\ket{1}_{R_{\rm C}}$ in this order transforms the state as
    \begin{align}
        &\ket{x}_{R_{\rm S}}\ket{\Delta x}_{R_{\rm M}}\ket{1}_{R_{\rm C}} \nonumber \\
        &\xrightarrow{F} \ket{x+\Delta x}_{R_{\rm S}}\ket{\Delta x}_{R_{\rm M}}\ket{1}_{R_{\rm C}} \nonumber \\
        &\xrightarrow{S} \ket{x+\Delta x}_{R_{\rm S}}\ket{-\Delta x}_{R_{\rm M}}\ket{1}_{R_{\rm C}} \nonumber \\
        &\xrightarrow{F} \ket{x}_{R_{\rm S}}\ket{-\Delta x}_{R_{\rm M}}\ket{1}_{R_{\rm C}} \nonumber \\
        &\xrightarrow{S} \ket{x}_{R_{\rm S}}\ket{\Delta x}_{R_{\rm M}}\ket{1}_{R_{\rm C}}.
    \end{align}
    Thus, we see that $SFSF$ acts as $I$ for any state in the form of $\ket{x}_{R_{\rm S}}\ket{\Delta x}_{R_{\rm M}}\ket{\phi}_{R_{\rm C}}$, where $\ket{\phi}_{R_{\rm C}}$ is any state on $R_{\rm C}$.
    This means that $SFSF=I$ and thus $(SF)^\dagger=SF$ on $\mathcal{A}+\mathcal{B}$.
\end{proof}

Then, combining these lemmas, we can prove Theorem \ref{th:phasegapOurs}.

\begin{proof}[Proof of Theorem \ref{th:phasegapOurs}]

    Combining Theorem \ref{th:Szegedy} with Lemmas \ref{lem:SameEigW} and \ref{lem:AdagB}, we see that $R_\mathcal{A}R_\mathcal{B}$ has eigenvalue 1 with multiplicity 1 and that any other eigenvalue is -1 or in the form of $\exp(\pm 2i \theta_l)$, where $\theta_1,\theta_2,...\in\left(0,\frac{\pi}{2}\right)$ are written as $\theta_l=\arccos |\lambda_l|$ with $\{\lambda_l\}$, the eigenvalues of $W$ with modulus less than 1.

    On the other hand, $R_\mathcal{A}$ and $R_\mathcal{B}$ are now
    \begin{equation}
        R_\mathcal{A} = 2\sum_{x\in\Omega} \ket{x}_{R_{\rm S}}\ket{0}_{R_{\rm M}}\ket{0}_{R_{\rm C}}\bra{x}_{R_{\rm S}}\bra{0}_{R_{\rm M}}\bra{0}_{R_{\rm C}}-I
    \end{equation}
    and
    \begin{widetext}
    \begin{align}
        R_\mathcal{B}
        &= 2\sum_{x\in\Omega} V^\dagger B^\dagger SFBV\ket{x}_{R_{\rm S}}\ket{0}_{R_{\rm M}}\ket{0}_{R_{\rm C}}\bra{x}_{R_{\rm S}}\bra{0}_{R_{\rm M}}\bra{0}_{R_{\rm C}} (V^\dagger B^\dagger SFBV)^\dagger - I \nonumber \\
        &= V^\dagger B^\dagger SFBV R_\mathcal{A} V^\dagger B^\dagger SFBV, \label{eq:RB}
    \end{align}
    \end{widetext}
    respectively.
    In Eq. (\ref{eq:RB}), we used Lemma \ref{lem:SF}.
    Note that, on $\mathcal{A}+\mathcal{B}$, $R$ in (\ref{eq:R}) acts as $R_\mathcal{A}$ and $V^\dagger B^\dagger SFBV R V^\dagger B^\dagger SFBV$ acts as $R_\mathcal{B}$.
    Thus, $RV^\dagger B^\dagger SFBV R V^\dagger B^\dagger SFBV=U^2$ acts as $R_\mathcal{A}R_\mathcal{B}$.
    Therefore, on $\mathcal{A}+\mathcal{B}$, the eigenvalues of $U$ are equal to the square root of those of $R_\mathcal{A}R_\mathcal{B}$.
    They include 1 or -1 with multiplicity 1, and, because of Lemma \ref{lem:Peig}, it is in fact 1 with the corresponding eigenstate $\ket{P}$.
    Any other eigenvalue of $U_W$ is $e^{\pm i\theta_l}$, $-e^{\pm i\theta_l}=e^{i(\pm\theta_l+\pi)}$, or $\pm i=e^{\pm \frac{\pi}{2}i}$, whose phase has modulus no less than
    \begin{equation}
    \arccos \left(\max \{|\lambda_l|\}\right)=\arccos (1-\Delta)
    \end{equation}
    in any case.

\end{proof}

\section{Details of Quantum Monte Carlo integration \label{sec:QMCIDetail}}

First, let us recall Theorem 5 in \cite{Miyamoto2022}.

\begin{theorem}[Theorem 5 in \cite{Miyamoto2022}, modified]
    Let $M\in\mathbb{N}$ and $\mathcal{X}$ be a set of $M$ real numbers, $X_0,...,X_{M-1}$, whose mean is $\mu:=\frac{1}{M}\sum_{i=0}^{M-1}X_i$ and sample variance satisfies $\frac{1}{M}\sum_{i=0}^{M-1}X_i^2 - \mu^2 \le \sigma^2$ with some $\sigma\in\mathbb{R}_+$.
    Suppose that we are given an access to a unitary operator $O_X$ that acts as Eq. (\ref{eq:OX}) for any $i\in[M]_0$.
    Let $\epsilon\in(0,4\sigma)$ and $\delta\in(0,1)$.
    Then, we have an access to a unitary operator $\tilde{O}_{\mathcal{X},\epsilon,\delta,\sigma}^{\rm mean}$ that acts on a system of two registers as 
    \begin{equation}
        \tilde{O}_{\mathcal{X},\epsilon,\delta,\sigma}^{\rm mean}\ket{0}\ket{0}=\sum_{y\in \mathcal{Y}} \alpha_y\ket{\phi_y}\ket{y}. \label{eq:ObarX}
    \end{equation}
    Here, $\mathcal{Y}$ is a finite set of real numbers that includes a subset $\tilde{\mathcal{Y}}$ consisting of $\epsilon$-approximations of $\mu$, $\{\alpha_y\}_{y\in\mathcal{Y}}$ are complex numbers satisfying $\sum_{\tilde{y}\in\tilde{\mathcal{Y}}}|\alpha_{\tilde{y}}|^2\ge 1-\delta$, and $\{\ket{\phi_y}\}_{y\in\mathcal{Y}}$ are states on the first register.
    In $\tilde{O}_{\mathcal{X},\epsilon,\delta,\sigma}^{\rm mean}$, queries to $O_{\mathcal{X}}$ are made, whose number is of order (\ref{eq:compQMCI}).
    $\tilde{O}_{\mathcal{X},\epsilon,\delta,\sigma}^{\rm mean}$ uses qubits whose number is of order (\ref{eq:qubitQMCI}).
    \label{th:QMCIorg}
\end{theorem}

We construct $O_{\mathcal{X},\epsilon,\delta,\sigma}^{\rm mean}$ in Theorem \ref{th:QMCI} using $\tilde{O}_{\mathcal{X},\epsilon,\delta,\sigma}^{\rm mean}$ and adding some operations afterward.

\begin{proof}[Proof of Theorem \ref{th:QMCI}]

Any $x\in\mathbb{R}$ can be written as $x=\sum_{i=-\infty}^{i=\infty}x_i2^i$, where $\{x_i\}_{i\in\mathbb{Z}}$ are binaries (0 or 1), that is, the binary representation of $x$.
We call $x_i$ the $i$th bit of $x$.
For $x\in\mathbb{R}$ and $a\in\mathbb{Z}$, we define
\begin{equation}
    \left\lfloor x \right\rfloor_{a} := \sum_{i=a}^{\infty} 2^ix_i.
\end{equation}
Namely, $\left\lfloor x \right\rfloor_{a}$ is the rounding of $x$ at the $a$th bit. 
We denote by $O^{\rm round}_a$ the operator for rounding: $O^{\rm round}_a\ket{x}\ket{0}=\ket{x}\Ket{\left\lfloor x \right\rfloor_{a}}$.
This is simply implemented by copying the higher-order qubits in the first register to the second register with CNOT gates.

Then, we can perform the following operation:
\begin{align}
& \ket{0}\ket{0}\ket{0} \nonumber \\
\rightarrow & \sum_{y\in \mathcal{Y}} \alpha_y\ket{\phi_y}\ket{y}\ket{0} \nonumber \\
\rightarrow & \sum_{y\in \mathcal{Y}} \alpha_y\ket{\phi_y}\ket{y}\ket{\left\lfloor y \right\rfloor_{b}}=:\ket{\Phi}.
\end{align}
Here, we use $\tilde{O}_{\mathcal{X},\epsilon^\prime,\delta^\prime,\sigma}^{\rm mean}$ at the first arrow and $O^{\rm round}_b$ at the second arrow, where $b=\left\lfloor \log_2\epsilon \right\rfloor$, $\epsilon^\prime:=2^{b-1}$ and $\delta^\prime:=\delta/4$.
$\mathcal{Y}$ is a finite set of real numbers that has a subset $\tilde{\mathcal{Y}}$ consisting of $\epsilon^\prime$-approximations of $\mu$ and the complex numbers $\{\alpha_y\}_{y\in\mathcal{Y}}$ satisfies $\sum_{\tilde{y}\in\tilde{\mathcal{Y}}}|\alpha_{\tilde{y}}|^2\ge 1-\delta^\prime$.
Note that, for any $\epsilon^\prime$-approximation $y$ of $\mu$, $y_b,y_{b+1},...$ and $\mu_b,\mu_{b+1},...$ are equal respectively, since any discrepancy in the $b$th or higher-order bits means that $|y-\mu|\ge 2^b >\epsilon^\prime$.
Thus, we have
\begin{equation}
    \ket{\Phi} = \left(\sum_{y\in \tilde{\mathcal{Y}}} \alpha_y\ket{\phi_y}\ket{y}\right)\otimes\ket{\left\lfloor \mu \right\rfloor_{b}} + \sum_{y\in \mathcal{Y} \setminus \tilde{\mathcal{Y}}} \alpha_y\ket{\phi_y}\ket{y}\otimes\ket{\left\lfloor y \right\rfloor_{b}}.
\end{equation}
Therefore, letting $\Ket{\tilde{\Phi}}:=\left(\sum_{y\in \mathcal{Y}} \alpha_y\ket{\phi_y}\ket{y}\right)\otimes\ket{\left\lfloor \mu \right\rfloor_{b}}$, we have
\begin{align}
    & \left\| \Ket{\tilde{\Phi}} -\ket{\Phi}\right\|\nonumber \\
    = & \left\| \sum_{y\in \mathcal{Y} \setminus \tilde{\mathcal{Y}}} \alpha_y\ket{\phi_y}\ket{y}\ket{\left\lfloor \mu \right\rfloor_{b}}- \sum_{y\in \mathcal{Y} \setminus \tilde{\mathcal{Y}}} \alpha_y\ket{\phi_y}\ket{y}\otimes\ket{\left\lfloor y \right\rfloor_{b}}\right\|\nonumber \\
    \le& \left\| \sum_{y\in \mathcal{Y} \setminus \tilde{\mathcal{Y}}} \alpha_y\ket{\phi_y}\ket{y}\ket{\left\lfloor \mu \right\rfloor_{b}}\right\|+\left\| \sum_{y\in \mathcal{Y} \setminus \tilde{\mathcal{Y}}} \alpha_y\ket{\phi_y}\ket{y}\otimes\ket{\left\lfloor y \right\rfloor_{b}}\right\|\nonumber \\
    \le & 2\sqrt{\sum_{y\in \mathcal{Y} \setminus \tilde{\mathcal{Y}}} |\alpha_y|^2} \nonumber \\
    \le & 2\sqrt{\delta^\prime} \nonumber \\
    =& \sqrt{\delta}.
\end{align}
This means that we can write
\begin{eqnarray}
    \ket{\Phi} &=& \Ket{\tilde{\Phi}}+\gamma \Ket{\tilde{\psi}} \nonumber \\
    &=& \left(\sum_{y\in \mathcal{Y}} \alpha_y\ket{\phi_y}\ket{y}\right)\otimes\ket{\left\lfloor \mu \right\rfloor_{b}}+\gamma \Ket{\tilde{\psi}},
\end{eqnarray}
where $\gamma:=\left\|\ket{\Phi}- \Ket{\tilde{\Phi}}\right\|\le\sqrt{\delta}$ and $\Ket{\tilde{\psi}}:=\frac{1}{\gamma}\left(\ket{\Phi}-\Ket{\tilde{\Phi}}\right)$.
Then, performing $\left(\tilde{O}_{\mathcal{X},\epsilon^\prime,\delta^\prime,\sigma}^{\rm mean}\right)^\dagger$ on the first and second register transforms $\ket{\Phi}$ to
\begin{equation}
    \ket{0}\ket{0}\ket{\left\lfloor \mu \right\rfloor_{b}}+\gamma \ket{\psi},
\end{equation}
where $\ket{\psi}$ is a state on the entire system.
Since $\left\lfloor \mu \right\rfloor_{b}$ is an $\epsilon$-approximation of $\mu$, we see that the above operation yields a state in the form of Eq. (\ref{eq:OmeanX}), with the first and second registers together seen as $R_1$ and the third one seen as $R_2$.

The number of queries to $O_\mathcal{X}$ in the entire process is that in $\tilde{O}_{\mathcal{X},\epsilon^\prime,\delta^\prime,\sigma}^{\rm mean}$ and $\left(\tilde{O}_{\mathcal{X},\epsilon^\prime,\delta^\prime,\sigma}^{\rm mean}\right)^\dagger$, that is, the double of that in $\tilde{O}_{\mathcal{X},\epsilon^\prime,\delta^\prime,\sigma}^{\rm mean}$, which is of order (\ref{eq:compQMCI}) since $\epsilon^\prime=\Theta(\epsilon)$ and $\delta^\prime=\Theta(\delta)$.
The number of qubits used in the entire process is also same as $\tilde{O}_{\mathcal{X},\epsilon^\prime,\delta^\prime,\sigma}^{\rm mean}$, and is of order (\ref{eq:qubitQMCI}).

\end{proof}

\section{Proof of Lemma \ref{lem:SpGapPert} \label{sec:proofLemSpGapPert}}

We use the following theorem.

\begin{theorem}[Theorem IIIa in \cite{BauerFike}]
    Let $n\in\mathbb{N}$ and $B,\tilde{B}\in\mathbb{C}^{n\times n}$.
    Assume that $B$ is diagonalizable and denote by $Q$ the matrix that diagonalizes $B$: $Q^{-1}BQ$ is diagonal.
    Denote by $\kappa$ the condition number of $Q$.
    Then, for each eigenvalue $\lambda$ of $B$, there exists an eigenvalue $\tilde{\lambda}$ of $\tilde{B}$ that satisfies
    \begin{equation}
        \left|\tilde{\lambda}-\lambda\right|\le \kappa \left\|B-\tilde{B}\right\|.
    \end{equation}
    \label{th:eigenPert}
\end{theorem}

We also use the following lemma.

\begin{lemma}
    Define $A$ as Eq. (\ref{eq:accRatio}) with $P$ in the form of Eq. (\ref{eq:PeL}), and $\tilde{A}$ as Eq. (\ref{eq:Atil}).
    Then, if $\epsilon:=\max_{x\in\Omega} |\tilde{L}(x)-L(x)|\le \frac{1}{4}$, 
    \begin{equation}
    |\tilde{A}(x,y)-A(x,y)|\le 8\epsilon \label{eq:Aerr}
    \end{equation}
    holds for any $x,y\in\Omega$.
    \label{lem:Aerr}
\end{lemma}

\begin{proof}

We consider the following two cases.

\ \\

\noindent (i) for $x,y\in\Omega$ such that $\frac{P(y)T(y,x)}{P(x)T(x,y)}\le 2$

Note that
\begin{equation}
    |e^a-1| \le 2|a|
\end{equation}
holds for any $a\in[-1,1]$.
Since 
\begin{eqnarray}
    &&\left|\left(\tilde{L}(x)-L(x)\right)-\left(\tilde{L}(y)-L(y)\right)\right| \nonumber \\ 
    &\le& \left|\tilde{L}(x)-L(x)\right|+\left|\tilde{L}(y)-L(y)\right| \nonumber \\
    &\le& 2\epsilon \nonumber \\
    &\le& 1, \label{eq:LyLxDiffErr}
\end{eqnarray}
we have
\begin{eqnarray}
    &&\left|\frac{P_0(y)e^{-\tilde{L}(y)}T(y,x)}{P_0(x)e^{-\tilde{L}(x)}T(x,y)}-\frac{P(y)T(y,x)}{P(x)T(x,y)}\right| \nonumber \\
    &=& \left|\frac{P(y)T(y,x)}{P(x)T(x,y)}e^{\left(\tilde{L}(x)-L(x)\right)-\left(\tilde{L}(y)-L(y)\right)}-\frac{P(y)T(y,x)}{P(x)T(x,y)}\right| \nonumber \\
    &\le& 2 \left|e^{\left(\tilde{L}(x)-L(x)\right)-\left(\tilde{L}(y)-L(y)\right)}-1\right| \nonumber \\
    &\le& 8\epsilon.
\end{eqnarray}
Since $\min\{1,\cdot\}$ is a $1$-Lipschitz function on $\mathbb{R}$, we obtain Eq. (\ref{eq:Aerr}).

\ \\

\noindent \noindent (ii) for $x,y\in\Omega$ such that $\frac{P(y)T(y,x)}{P(x)T(x,y)}> 2$

Because of Eq. (\ref{eq:LyLxDiffErr}),
\begin{equation}
    e^{\left(\tilde{L}(y)-L(y)\right)-\left(\tilde{L}(x)-L(x)\right)} \ge e^{-2\epsilon} \ge e^{-\frac{1}{2}} \ge \frac{1}{2} \label{eq:temp}
\end{equation}
holds, which means
\begin{align}
    \frac{P_0(y)e^{-\tilde{L}(y)}T(y,x)}{P_0(x)e^{-\tilde{L}(x)}T(x,y)}&=\frac{P(y)T(y,x)}{P(x)T(x,y)}e^{\left(\tilde{L}(x)-L(x)\right)-\left(\tilde{L}(y)-L(y)\right)}\nonumber \\
    &\ge 1
\end{align}
and thus $\tilde{A}(x,y)=1=A(x,y)$.

\ \\

Thus, in both cases, Eq. (\ref{eq:Aerr}) holds.
\end{proof}

Then, Lemma \ref{lem:SpGapPert} is proven as follows.

\begin{proof}[Proof of Lemma \ref{lem:SpGapPert}]

    Theorem \ref{th:eigenPert} implies that
    \begin{equation}
        \tilde{\Delta} \ge \Delta - \kappa \|\delta W\|, \label{eq:SpGapdelW}
    \end{equation}
    where $\delta W:=\tilde{W}-W$ and $W$ (resp. $\tilde{W}$) is the transition matrix of $\mathcal{C}_L$ (resp. $\mathcal{C}_{\tilde{L}}$).

    Then, let us bound $\|\delta W\|$.
    To do this, we use a well-known inequality \cite[Corollary 2.3.2]{golub2013matrix}:
    \begin{equation}
        \|\delta W\|\le\sqrt{\|\delta W\|_1\|\delta W\|_\infty}. \label{eq:2normBound}
    \end{equation}
    We also have
    \begin{align}
        \|\delta W\|_\infty 
        &= \max_{x\in\Omega} \sum_{y\in\Omega} |\delta W_{x,y}| \nonumber \\
        &= \max_{x\in\Omega} \left(|\delta W_{x,x}| + \sum_{y\in\Omega\setminus\{x\}} |\delta W_{x,y}|\right)\nonumber \\
        &=\max_{x\in\Omega} \left(\left|\sum_{y\in\Omega\setminus\{x\}} T(x,y)\left(A(x,y)-\tilde{A}(x,y)\right)\right|+ \right. \nonumber \\
        & \qquad \left.\sum_{y\in\Omega\setminus\{x\}} \left|T(x,y)\left(\tilde{A}(x,y)-A(x,y)\right)\right|\right) \nonumber \\
        &\le 2\max_{x\in\Omega} \sum_{y\in\Omega\setminus\{x\}}T(x,y)\left|\tilde{A}(x,y)-A(x,y)\right| \nonumber \\
        &\le 16\epsilon \max_{x\in\Omega}\sum_{y\in\Omega\setminus\{x\}}T(x,y) \nonumber \\
        &\le 16\epsilon, \label{eq:maxNormBound}
    \end{align}
    where we used Lemma \ref{lem:Aerr} at the second inequality.
    Similarly, we have
    \begin{align}
        \|\delta W\|_1 
        &= \max_{y\in\Omega} \sum_{x\in\Omega} |\delta W_{x,y}| \nonumber \\
        &\le 16\epsilon \max_{y\in\Omega}\sum_{x\in\Omega\setminus\{y\}}T(x,y) \nonumber \\
        &\le 16\epsilon \max_{y\in\Omega}\sum_{x\in\Omega}T(x,y), \label{eq:1normBound}
    \end{align}
    Combining Eqs. (\ref{eq:2normBound}), (\ref{eq:maxNormBound}) and (\ref{eq:1normBound}) with Eq. (\ref{eq:SpGapdelW}), we obtain Eq. (\ref{eq:SpGapPert}).
\end{proof}

\section{Proof of Lemma \ref{lem:TVvsLerr} \label{sec:proofLemTVvsLerr}}

\begin{proof}[Proof of Lemma \ref{lem:TVvsLerr}]

If $\epsilon:=\max_{x\in\Omega} |\tilde{L}(x)-L(x)|\le \frac{1}{4}$, Eq. (\ref{eq:Aerr}) holds for any $x,y\in\Omega$ because of Lemma \ref{lem:Aerr}.
Combining this with Eq. (\ref{eq:AppARDistDiffMH}), we obtain Eq. (\ref{eq:diffPtilP}).

If $\epsilon>\frac{1}{4}$, Eq. (\ref{eq:diffPtilP}) holds trivially since the RHS is larger than 1 and the LHS is not larger than 1 by definition.
    
\end{proof}

\bibliographystyle{apsrev4-1}
\bibliography{reference}

\end{document}